\documentclass[11pt,a4paper]{amsart}
\usepackage[utf8]{inputenc}
\usepackage[a4paper,margin=2.5cm]{geometry}
\usepackage{amsmath,amssymb,amsthm,mathtools,abraces}
\usepackage{stmaryrd}
\usepackage{color,xcolor}
\usepackage[hypertexnames=false,hyperfootnotes=false,colorlinks=true,linkcolor=blue,%
citecolor=purple,filecolor=magenta,urlcolor=cyan,unicode,linktocpage=true,pagebackref=false]{hyperref}
\usepackage{yhmath}
\usepackage{verbatim}

\newcommand\be{\begin{equation}}
\newcommand\ee{\end{equation}}

\newcommand\p{\partial}

\newcommand{\Normord}[1]{:\mathrel{#1}:}
\DeclareMathOperator{\sppan}{span}

\def\lvac{\left <0\right |}
\def\rvac{\left |0\right >}

\DeclareMathOperator{\res}{Res}
\DeclareMathOperator{\odd}{odd}

\DeclareMathOperator{\DP}{SP}
\DeclareMathOperator{\OP}{OP}

\DeclareMathOperator{\KP}{KP}

\makeatletter
\@namedef{subjclassname@2020}{%
  \textup{2020} Mathematics Subject Classification}
\makeatother

\newtheorem{theorem}{Theorem}[section]
\newtheorem{lemma}{Lemma}[section]
\newtheorem{proposition}[lemma]{Proposition}
\newtheorem{corollary}[lemma]{Corollary}
\newtheorem{remark}{Remark}[section]

\newtheorem*{theorem*}{Theorem}

\newtheorem{definition}{Definition}[section]

\numberwithin{equation}{section}

\usepackage{filecontents}

\title[Elements of spin Hurwitz theory]{Elements of spin Hurwitz theory: closed algebraic formulas, blobbed topological recursion, and a proof of  the Giacchetto--Kramer--Lewa\'nski conjecture}

\author{Alexander Alexandrov}

\address[A. Alexandrov]{Center for Geometry and Physics, Institute for Basic Science (IBS), Pohang 37673, Korea
}

\email{\tt alex@ibs.re.kr}

\author{Sergey Shadrin}

\address[S. Shadrin]{Korteweg-de Vriesinstituut voor Wiskunde, 
	Universiteit van Amsterdam, Postbus 94248,
	1090GE Amsterdam, Nederland}
\email{\tt s.shadrin@uva.nl}


\begin{document}

\begin{abstract} 
In this paper, we discuss the properties of the generating functions of spin Hurwitz numbers.
In particular, for spin Hurwitz numbers with arbitrary ramification profiles, we construct the weighed sums which are given by
 Orlov's hypergeometric solutions of the 2-component BKP hierarchy. We derive the closed algebraic formulas for the correlation functions associated with these tau-functions, and under reasonable analytical assumptions we prove the loop equations (the blobbed topological recursion). Finally, we prove a version of topological recursion for the spin Hurwitz numbers with the spin completed cycles (a generalized version of the Giacchetto--Kramer--Lewa\'nski conjecture).
\end{abstract}

\maketitle

\tableofcontents 

\def\thefootnote{\arabic{footnote}}

\setcounter{equation}{0}

\section{Introduction}

\subsection{Topological recursion and integrability} It is well known that the Chekhov--Eynard--Orantin topological recursion~\cite{EynardOrantin} is closely related to integrability. However, the details of a general relationship between the two phenomena remain unclear. Although topological recursion is believed to be a universal property for a huge class of enumerative geometry and physics problems, the proofs of its validity are often 
model-dependent and technically involved, and at this point, despite the lack of understanding of the general relationship, various universal properties of integrability often 
help to prove topological recursion.

One of the most general applications of integrability to topological recursion is given by the 
weighted Hurwitz numbers. The generating functions of the weighted Hurwitz numbers are hypergeometric tau-functions of the $2$-component KP ($2$-KP) hierarchy.
The study of topological recursion for the general hypergeometric solutions of the $2$-KP hierarchy was initiated in \cite{ACEH1,ACEH3} (subsuming a huge list of particular examples known before). Many elements of the general construction including quantum and classical spectral curves were properly identified there. However, the topological recursion was proved only for an infinite-dimensional family of solutions with polynomial weight functions and finite sets of the second times of the $2$-KP hierarchy. Topological recursion for the much more general families of the hypergeometric solutions of the $2$-KP hierarchy was proved in \cite{BDKS1,BDKS2}. The proof there is based on the free field description of the KP hierarchy, more specifically, on the free fermion construction and the boson-fermion correspondence.

These results indicate that the same line of reasoning can be applied to any integrable hierarchy with free fermion description. 
In this paper, we describe topological recursion for the hypergeometric solutions of the $2$-component 
BKP ($2$-BKP) hierarchy. A well-known neutral fermion description of the BKP hierarchy allows us to follow the general approach for the $2$-KP hierarchy, derived in \cite{BDKS1,BDKS2}. Many steps can be repeated without essential changes, but some specifics of the $2$-BKP case (mostly important, the built-in oddness of the parametrizations) require extra analysis and lead to new phenomena.
We derive the general closed algebraic formulas for the correlation functions in the $2$-BKP case and prove the blobbed topological recursion~\cite{BorotShadrin}, that is, the linear and quadratic loop equations. 

\subsection{BKP and spin Hurwitz theory} The BKP hierarchy is believed to govern the spin Hurwitz numbers in essentially the same way as the KP hierarchy governs the ordinary Hurwitz numbers \cite{MMNQ}. However, the important construction of the {\em weighted spin Hurwitz numbers} (in the sense of~\cite{GPH}) is still unavailable in the literature. In this paper we show how to construct integrable generating functions of spin Hurwitz numbers for arbitrary ramification profiles and number of the branch points. These generating functions are Orlov's hypergeometric tau-functions of the $2$-component BKP hierarchy~\cite{OBKP}, and the weights associated with the ramifications serve as parameters. It is not clear at the moment how to reduce naturally the number of parameters and to define the direct analogs of weighted Hurwitz numbers in the spin case. To this end, we suggest two possible candidates for the elementary weight functions. 

It is well known that the tau-functions of the KP and BKP hierarchies are related to each other by a simple quadratic relation \cite{JMBKP}. Following Orlov \cite{OBKP}, we describe this relation for the hypergeometric tau-functions. Namely, for any hypergeometric tau-function of the $2$-BKP hierarchy we find the corresponding tau-function of the $2$-KP hierarchy. It is easy to see that such KP tau-function is not unique. This relation between tau-functions should lead to the non-trivial relations between the spin and ordinary Hurwitz numbers. 

\subsection{Giacchetto--Kramer--Lewa\'nski conjecture and its generalization}

Additional input and motivation to study the correlation functions of the corresponding hypergeometric $2$-BKP tau-functions comes from a recent work of Giacchetto, Kramer, and Lewa\'nski~\cite{GKL}. They study in detail the theory of so-called spin Hurwitz numbers with completed cycles, both single and double, whose elements occur naturally in a number of other works in relation to computation of the volumes of strata in the moduli spaces of holomorphic differentials~\cite{EOP} and Gromov--Witten theory of K\"ahler surfaces~\cite[Introduction]{LeeRoot}. 

Remarkably, Giacchetto, Kramer, and Lewa\'nski propose a conjectural statement on $\mathbb{Z}_2$-equi\-va\-riant version of topological recursion for the correlation differentials of these numbers, and they prove that the statement on topological recursion is equivalent to an ELSV-type formula for spin Hurwitz numbers with completed cycles that expresses these numbers in terms of the Chiodo classes twisted by the Witten $2$-spin class.

Using the formulas for the correlation functions and loop equations we prove a natural generalization of the Giacchetto--Kramer--Lewa\'nski conjecture, that is, a $\mathbb{Z}_2$-equivariant version of topological recursion for the double spin Hurwitz numbers with arbitrary finite linear combinations of the spin completed cycles. 

Let us remark that while with the motivation coming from~\cite{GKL} we focus on this particular family of spin Hurwitz numbers, we expect that our modification of the methods of~\cite{BDKS1,BDKS2} should immediately work for other families of the generating functions of the spin Hurwitz numbers, analogous to the families investigated in \cite{BDKS2}. We also expect that the integrable approach to the topological recursion in the BKP case should be as universal as for the KP case. Moreover, without significant modifications, it should also work for other integrable hierarchies described by free fermions.


\subsection{Notation} A partition $\lambda$ is {\em strict}, if $\lambda_1>\lambda_2>\lambda_3>\dots>\lambda_{\ell(\lambda)}>\lambda_{\ell(\lambda)+1}=0$,  where $\ell(\lambda)$ is the {\em length} of the partition. We denote the set of strict partitions, including the empty one, by $\DP$. A partition $\lambda$ is {\em odd} if all parts in $\lambda$ are odd. We denote the set of odd partitions, including the empty one, by $\OP$. For a partition $\lambda$ by $\lambda(k)$ we denote the number of parts equal to $k$.

\subsection{Organization of the paper} In Section \ref{S2} we recall the neutral fermion description of the BKP hierarchy. In Section \ref{S3} we explain how to construct the generating functions of the spin Hurwitz numbers that solve the 2-BKP hierarchy and how these tau-functions can be related to the generating functions of the ordinary Hurwitz numbers. Section \ref{S4} is devoted to the correlation functions for the general hypergeometric tau-functions of the 2-BKP hierarchy. In Section \ref{S5} we prove that these correlation functions, under mild analytic assumptions, satisfy linear and quadratic loop equations. In Section \ref{S6} we derive explicit expressions for the $n$-point correlation functions. In Section \ref{S7} we use these expressions to prove the topological recursion for the spin Hurwitz numbers with the spin completed cycles.

\subsection{Acknowledgments}
	
The work of A.~A. was supported by the Institute for Basic Science (IBS-R003-D1). The work of S.~S. was supported by the Netherlands Organization for Scientific Research.

The authors thank A.~Giacchetto, R.~Kramer, and D.~Lewa\'nski for useful discussions and an anonymous referee for the suggested improvements.


\section{Neutral fermions and boson-fermion correspondence}\label{S2}

In this section we remind the reader the neutral fermion formalism and boson-fermion correspondence in the framework of the BKP hierarchy.
More detailed descriptions can be found in \cite{JMBKP,You,vdLASM,OBKP}.


\subsection{Neutral fermions}

Let $\phi_k$, $k\in {\mathbb Z}$, be the neutral free fermions satisfying the canonical anticommutation relations
\be\label{ac}
\left\{\phi_k,\phi_m\right\}=(-1)^k\delta_{k+m,0}.
\ee
Note that $\phi_0^2=1/2$. These relations define the Clifford algebra as an associative algebra.

For the vacuum vector $\rvac$ and the co-vacuum $\lvac$, satisfying
\be\label{vp1}
\phi_m \rvac =0,\,\,\,\,\, \lvac \phi_{-m}=0,\,\,\,\,\, m<0
\ee
the elements $\phi_{k_1}\phi_{k_2}\dots \phi_{k_m}\rvac$  with $k_1>k_2>\dots>k_m\geq0$ form a basis of the {\em neutral fermion Fock space} ${\mathcal F}_B$,
\be\label{bas1}
{\mathcal F}_B =\sppan\left\{ \phi_{k_1}\phi_{k_2}\dots \phi_{k_m}\rvac \ |\ k_1>k_2>\dots>k_m\geq 0 \right\}, 
\ee
and its dual
\be\label{bas2}
{\mathcal F}_B^* =\sppan\left\{\lvac  \phi_{k_m}\dots  \phi_{k_2}\phi_{k_1} \ |\  k_1<k_2<\dots<k_m\leq0  \right\}.
\ee
The space ${\mathcal F}_B$ splits into two subspaces
\be
{\mathcal F}_B={\mathcal F}_B^0\oplus {\mathcal F}_B^1,
\ee
where ${\mathcal F}_B^0$ and ${\mathcal F}_B^1$ denote the subspaces with even and odd numbers of generators $\phi_k$, respectively. The same decomposition exists for ${\mathcal F}_B^*$.

There is a nondegenerate bilinear pairing ${\mathcal F}_B \times {\mathcal F}_B^* \rightarrow {\mathbb C}$, and the pairing of $\left<U\right | \in {\mathcal F}_B^*$ and $ \left|V\right>\in {\mathcal F}_B$ is denoted by $\left<U |V\right>$ with
\be
\left<0|0\right>=1.
\ee
 The {\em vacuum expectation values} of an element $a$ of the Clifford algebra is a pairing of $\lvac$ and $a\rvac$, which is denoted by  $\lvac a \rvac$. It is uniquely defined by the anticommutation relations (\ref{ac}), property (\ref{vp1}), and the following relation:
 \be
 \lvac \phi_0 \rvac =0. 
 \ee
In particular, if $a$ is an odd element of the Clifford algebra, then $\lvac a \rvac =0$. It is easy to see that the bases in (\ref{bas1}) and  (\ref{bas2}) are orthogonal. 
Let us focus on the space ${\mathcal F}_B^0$ and its dual. The basis can be labelled by strict partitions $\lambda \in \DP$ in the following way: 
\be\label{lambda}
\left|\lambda\right> = 
\begin{cases}
 \phi_{\lambda_1}\phi_{\lambda_2}\dots \phi_{\lambda_{\ell(\lambda)}}\rvac     & \mathrm{for}  \,\,\, \ell(\lambda)=0 \mod 2,\\
\sqrt{2} \phi_{\lambda_1}\phi_{\lambda_2}\dots \phi_{\lambda_{\ell(\lambda)}}\phi_0 \rvac  & \mathrm{for}   \,\,\, \ell(\lambda)=1 \mod 2,
\end{cases}
\ee
and similarly for ${\mathcal F}_B^{0*}$. From the anticommutation relations we have
\be
\left<\lambda|\mu\right>=(-1)^{|\lambda|} \delta_{\lambda,\mu}.
\ee

It is easy to see that
\be
\lvac \phi_k\phi_{m} \rvac =\delta_{k+m,0}H[m], 
\ee
where
\be\label{Ha}
H[m]  =
\begin{cases}
\displaystyle{0}  & \mathrm{for} \quad m<0,\\
\displaystyle{\frac{1}{2}}  & \mathrm{for} \quad m=0,\\
\displaystyle{(-1)^m} \,\,\,\,\,\, & \mathrm{for} \quad m>0.
\end{cases}
\ee

Bilinear combinations of neutral fermions $\phi_k\phi_m$ satisfy the commutation relations of the Lie algebra $B_\infty$. Let $\left(E_{i,j}\right)_{k,l}=\delta_{i,k}\delta_{j,l}$ be the standard basis of the matrix units $\left.\left\{E_{{i,j}}\right| i,j \in {\mathbb Z}\right\}$. Then $\phi_k\phi_{-m}$ corresponds \cite{JMBKP} to 
\be
F_{k,m}=(-1)^m E_{k,m}-(-1)^k E_{-m,-k}
\ee
with the commutation relations
\be\label{comF}
\left[F_{a,b},F_{c,d}\right]=(-1)^b\delta_{b,c}F_{a,d}-(-1)^a\delta_{a+c,0}F_{-b,d}+(-1)^b\delta_{b+d,0}F_{c,-a}-(-1)^a\delta_{a,d}F_{c,b}.
\ee

For the bilinear combinations of neutral fermions we introduce the {\em normal ordering} by 
\begin{equation}
\Normord{\phi_k\phi_{m}}=\phi_k\phi_{m}-\lvac \phi_k\phi_{m} \rvac.
\end{equation} 
It is skewsymmetric 
\be
\Normord{\phi_k\phi_{m}}=-\Normord{\phi_m\phi_{k}},
\ee
in particular, $\Normord{\phi_k\phi_{k}}=0$.
The normal ordered quadratic combinations of neutral fermions satisfy the commutation relations of a central extension of the algebra $B_\infty$:
\begin{align}
[\Normord{\phi_a\phi_b},\Normord{\phi_c\phi_d}] & = (-1)^b \delta_{b+c,0}\Normord{\phi_a\phi_d} -(-1)^a\delta_{a+c,0}\Normord{\phi_b\phi_d}\\ \notag 
& \quad + (-1)^b\delta_{b+d,0}\Normord{\phi_c\phi_a} -(-)^a\delta_{a+d,0}\Normord{\phi_c\phi_b} \\ \notag & \quad + (\delta_{c,b}\delta_{a,d}-\delta_{a-c,0}\delta_{b-d,0})( (-1)^{a}H[b]- (-1)^{b}H[a]),
\end{align}
where $H[a]$ is given by (\ref{Ha}). The operator $\Normord{\phi_k\phi_{-m}}$ corresponds to the projective representation of the Lie algebra $B_\infty$, and will also be denoted by $\hat{F}_{k,m}$.

Let us consider the generating function
\be
\phi(z)=\sum_{k\in {\mathbb Z}}\phi_k z^k. 
\ee
It satisfies the anti-commutation relation 
\be\label{phid}
\left\{\phi(z),\phi(w)\right\}=\delta(z+w).
\ee
Here we introduce the delta-function
\be\label{delta}
\delta(z-w)=\sum_{k\in {\mathbb Z}}\left(\frac{z}{w}\right)^k.
\ee
It satisfies
\be
\delta(z-w)f(z)=\delta(z-w)f(w)
\ee
for any formal series $f(z)\in {\mathbb C}[\![z,z^{-1}]\!]$ and can be represented as
\be
2\delta(z+w)=\iota_{|z|>|w|}\frac{z-w}{z+w}-\iota_{|w|>|z|}\frac{z-w}{w+z},
\ee
where $\iota_{|z|>|w|}$ is the operation of Laurent series expansion in the
region $|z|>|w|$.

Quadratic combinations of the generating functions $\phi(z)$ generate a Lie algebra with the following commutation relations
\begin{multline}\label{bilc}
\left[\phi_1(z_1)\phi(w_1),\phi(z_2)\phi(w_2)\right]=\delta(w_1+z_2)\phi(z_1)\phi(w_2)-\delta(z_1+z_2)\phi(w_1)\phi(w_2)\\
+\delta(w_1+w_2)\phi(z_2)\phi(z_1)-\delta(z_1+w_2)\phi(z_2)\phi(w_1).
\end{multline}
For the normal ordered operator we have
\be
\phi(z)\phi(w)=\Normord{\phi(z)\phi(w)}+\frac{1}{2}\iota_{|z|>|w|} \frac{z-w}{z+w}.
\ee


\subsection{Vertex operators}

For $k\in {\mathbb Z}_{\odd}$ we introduce {\em bosonic} operators
\be\label{bosop}
J_k=\frac{1}{2}\sum_{m\in {\mathbb Z}}(-1)^{m+1}\Normord{\phi_m\phi_{-m-k}}
\ee
satisfying a commutation relation of the Heisenberg algebra
\be\label{comJ}
\left[J_k,J_m\right]=\frac{k}{2}\delta_{k+m,0}.
\ee
From (\ref{vp1}) we have:
\be
J_m \rvac =0,\,\,\,\,\, \lvac J_{-m}=0,\,\,\,\,\, m>0.
\ee

Let us consider the {\em vertex operator} for the BKP hierarchy introduced in \cite{JMBKP},
\be
\widehat{V}^{(1)}_B(z)=\exp\left(\sum_{k\in{\mathbb Z}_{\odd}^+}z^k t_k\right)\exp\left(-2\sum_{k\in{\mathbb Z}_{\odd}^+}\frac{1}{kz^k}\frac{\p}{\p t_k}\right).
\ee
These operators satisfy the anticommutation relation
\be\label{Vac}
\left\{\widehat{V}^{(1)}_B(z),\widehat{V}^{(1)}_B(w)\right\}=2 \delta(z+w)
\ee
similar to the relation (\ref{phid}).

It is convenient to introduce generating functions of the bosonic operators:
\be
J_+({\bf t}) = \sum_{k\in {\mathbb Z}_{\odd}^+} t_k J_k,\,\,\,\,\,\,\, J_-({\bf s}) = \sum_{k\in {\mathbb Z}_{\odd}^+} s_k J_{-k}.
\ee
Then one has
\begin{equation}\label{Vtp}
\begin{split}
\widehat{V}^{(1)}_B(z)\lvac e^{J_+({\bf t})}&=2\lvac \phi_0  e^{J_+({\bf t})} \phi(z),\\
\widehat{V}^{(1)}_B(z)\lvac \phi_0 e^{J_+({\bf t})}&= \lvac e^{J_+({\bf t})} \phi(z).
\end{split}
\end{equation}

Let us consider a bilinear combination of the vertex operators
\be
\widehat{Y}_B(z,w)=\frac{1}{2}\widehat{V}^{(1)}_B(z) \widehat{V}^{(1)}_B(w).
\ee
Using the anti-commutation relation (\ref{Vac}) it is easy to show that the vertex operators $\widehat{Y}_B(z,w)$ satisfy a commutation relation equivalent to the relation (\ref{bilc}) for the bilinear combinations $\phi(z)\phi(w)$:
\begin{align}
\left[\widehat{Y}_B(z_1,w_1),\widehat{Y}_B(z_2,w_2)\right] & = \delta(w_1+z_2)\widehat{Y}_B(z_1,w_2)
-\delta(z_1+z_2)\widehat{Y}_B(w_1,w_2)
\\ \notag & \quad 
+\delta(w_1+w_2)\widehat{Y}_B(z_2,z_1)-\delta(z_1+w_2)\widehat{Y}_B(z_2,w_1).
\end{align}

It is also convenient to consider its regularized version, corresponding to $\Normord{\phi(z)\phi(-w)}$
\be
{\widehat V}^{(2)}_B(z,w)=\widehat{Y}_B(z,w)-\frac{1}{2}\iota_{|z|>|w|} \frac{z-w}{{z}+w},
\ee
where for the second term we assume the series expansion in $|z|>|w|$.
This expression has no pole at $z=-w$, moreover, it is antisymmetric with respect to the permutation of $z$ and $w$.
These vertex operators can be represented as
\be\label{Vr}
\widehat{V}^{(2)}_B(z,w)=\frac{1}{2}\frac{z-w}{z+w}\left(e^{ \sum_{k\in{\mathbb Z}_{\odd}^+}t_k (z^k +w^k)}e^{-2\sum_{k\in{\mathbb Z}_{\odd}^+}\left(\frac{1}{kz^k}+\frac{1}{kw^k}\right)\frac{\p}{\p t_k}}-1\right).
\ee 

From (\ref{Vtp}) it follows that
\be\label{ara}
{\widehat V}^{(2)}_B(z,w) \lvac e^{J_+({\bf t}) }=\lvac e^{J_+({\bf t}) } \Normord{\phi(z)\phi(w)}.
\ee

\subsection{Boson-fermion correspondence}

For the neutral fermions the boson-fermion correspondence describes an isomorphism \cite{You}
\be
\sigma_B^i:\,\,\,\,\, {\mathcal F}_B^i \simeq B^{(i)}={\mathbb C}[\![t_1,t_3,t_5\dots ]\!]
\ee
for $i=0,1$. Here
\be
\sigma_B^i (\left| i\right> )=1,
\ee
where we introduce $\left|1\right>= \sqrt{2}\phi_0 \rvac$, and for both $i=0,1$ we have
\be
\sigma_B^i  J_{-k} (\sigma_B^i)^{-1}=\frac{k}{2}  t_k ,\,\,\,\,\, \sigma_B^i  J_k (\sigma_B^i)^{-1} = \frac{\p}{\p t_k}
\ee
for $k\in{\mathbb Z}_{\odd}^+$. The boson-fermion correspondence is given by
\be
\sigma_B^i(\left| a \right>)=\begin{cases}
\displaystyle{\left<1\right| e^{J_+({\bf t}) } \left| a \right>} \,\,\,\,\,\,\,\,\,\,\,\,\,\,\,\,\, \mathrm{for} \quad  \left| a \right> \in  {\mathcal F}_B^1,\\[4pt]
\displaystyle{\lvac e^{J_+({\bf t}) } \left| a \right> }\,\,\,\,\,\,\,\,\,\,\,\,\,\,\,\,\, \mathrm{for} \quad \left| a \right> \in  {\mathcal F}_B^0,
\end{cases}
\ee
where $\left<1\right|=\sqrt{2}\lvac \phi_0$. The boson-fermion correspondence between two different representations of the central extension of the $B_\infty$ algebra is given by
\be
\sigma_B^i \Normord{\phi(z) \phi(w)} (\sigma_B^i )^{-1}={\widehat V}^{(2)}_B(z,w).
\ee
Below we will work only with ${\mathcal F}_B^0$ component of the fermionic Fock space and its bosonic counterpart. 


Relation between the Schur Q-functions and the BKP hierarchy, in particular, is described by the following result of You:
\begin{theorem}[\cite{You}]\label{TYou}
For the states (\ref{lambda}) the boson-fermion correspondence yields
\be
\sigma_{B}^0 (\left| \lambda \right> )= 2^{-\ell(\lambda)/2} Q_\lambda({\bf t}/2).
\ee
\end{theorem}
Here $Q_\lambda$ are the Schur Q-functions (see Section III.8 of \cite{Mac} for definition and details).

It was shown by Date, Jimbo, Kashiwara, and  Miwa \cite{JMBKP} that for any {\em group element} of the central extension of the algebra $B_\infty$,
\be
G=\exp\left(\sum_{k,m\in {\mathbb Z}} a_{km} \Normord{\phi_k \phi_{m}}\right),
\ee
the bosonic image of the fermionic state $G e^{J_-({\bf s})}\rvac $ solves the 2-BKP hierarchy. Namely,
\be
\tau({\bf t},{\bf s}) =\lvac e^{J_+({\bf t})} G e^{J_-({\bf s})}\rvac 
\ee
is a tau-function of 2-BKP hierarchy.

\section{Hypergeometric tau-functions and weighted spin Hurwitz numbers}\label{S3}

In this section we suggest a way to construct the weighed sums of the spin Hurwitz numbers which solve the $2$-BKP hierarchy. There is a certain ambiguity associated to the choice of the weights, and we discuss two natural candidates for the role of the elementary weight functions. 


\subsection{Hypergeometric tau-functions of 2-BKP hierarchy}

Following Orlov  \cite{OBKP} we consider a set of parameters $T_n$, $n\in {\mathbb Z}$, such that $T_{n}=-T_{-n}$. In particular, $T_0=0$. Then
\begin{equation}
		\sum_{k\in {\mathbb Z}}(-1)^{k} T_k \Normord{\phi_k \phi_{-k}}=2 \sum_{k\in {\mathbb Z}_+}  (-1)^{k}T_k   \Normord{\phi_k \phi_{-k}}=2 \sum_{k\in {\mathbb Z}_+} (-1)^{k}  T_k   \hat{F}_{k,k}.
\end{equation}
Consider the group element
\be\label{Ddiag}
{\mathcal D}=\exp\left(\sum_{k\in {\mathbb Z}}(-1)^{k} T_k \Normord{\phi_k \phi_{-k}}\right),
\ee
then
\be
\tau({\bf t},{\bf s})= \lvac e^{J_+({\bf t})}{\mathcal D} e^{J_-({\bf s})}\rvac
\ee
is a tau-function of the $2$-component BKP hierarchy symmetric in the variables $t_k$ and $s_k$. From Theorem \ref{TYou} it follows that this tau-function
has an equivalent description \cite{OBKP}
\be\label{hyperg}
\tau({\bf t},{\bf s}) =\sum_{\lambda \in \DP} 2^{-\ell(\lambda)} e^{2T_{\lambda_1}+\dots+2T_{\lambda_{\ell(\lambda)}}}Q_\lambda({\bf t}/2)Q_\lambda({\bf s}/2). 
\ee
Here $\DP$ is the set of all strict partitions including the empty one. These are the {\em hypergeometric tau-functions} of the 2-BKP hierarchy.

Consider $T(x)$, an odd function such that $T(k)=T_k$ for $k \in {\mathbb Z}$.  For future applications it is natural to introduce the topological expansion parameter $\hbar$. Consider a new \emph{even} function $\overline{\psi}(z)$ such that 
\be
\overline\psi(\hbar(z+1/2)) =  T(z+1)-T(z).
\ee
We assume that $\overline\psi(z)$ is itself a series in $\hbar^2$, $\overline\psi(\hbar^2,z) = \sum_{d=0}^\infty \hbar^{2g} \psi_{2d}(z)$, where $\psi_{2d}(z)$ is an even formal power series in $z$, therefore $T_k$ also depend on $\hbar$. The constant term of this series, $\psi_0=\overline\psi(0,z)$, is also denoted by $\psi=\psi(z)$.

\begin{remark}
Our definition of the parameters $T_k$ corresponds to the doubled parameters of \cite{OBKP} with the inverse sign.
\end{remark}

\subsection{Spin Hurwitz numbers}

Spin Hurwitz numbers, which count the ramified coverings with sign coming from spin structure, were introduces by Eskin, Okounkov, and Pandharipande \cite{EOP}. Using TQFT, Gunningham \cite{Guni} found a combinatorial expression for all genera spin Hurwitz numbers, which uses the representation theory of Sergeev's group. In this section we recall this combinatorial expression.
 We address the reader to \cite{EOP,Guni,LeeG,LeeRoot,MMNQ,GKL} for the basic definitions and properties. Different authors use different conventions, our notation is consistent with that of \cite{GKL}.


For any set of variables or parameters $r_k$ and any partition $\mu$ let us denote
\be
r_\mu=\prod_{j=1}^{\ell(\mu)} r_{\mu_j}.
\ee
Let $\OP(d)$ and $\DP(d)$ be the sets of odd partitions and strict partitions of the size $d$ respectively.
Then the Schur Q-functions can be expanded as
\be
Q_\lambda= 2^{\frac{\ell(\lambda)-\delta(\lambda)}{2}} \sum_{\mu \in \OP(|\lambda|)}  \frac{\zeta^\lambda_\mu}{z_\mu} p_\mu
\ee
with the inverse relation
\be\label{pasQ}
 p_\mu= 2^{-\ell(\mu)}  \sum_{\lambda \in \DP(|\mu|)} 2^{-\frac{\ell(\lambda)+\delta(\lambda)}{2}} \zeta^\lambda_\mu Q_\lambda.
\ee
Here $p_k=k t_k$ are the independent variables,
\be
\delta(\mu)= 
\begin{cases}
0,\,\,\,\,\,\, {\rm for} \,\,{\rm even}\,\,  \ell(\mu)\\ 
1,\,\,\,\,\,\, {\rm for}\,\, {\rm odd}\,\, \ell(\mu)
\end{cases}
\ee
and $z_\mu=\prod_k \mu(k)! k^{\mu(k)}$.

The characters of the Sergeev group $\zeta_\mu^\rho$ satisfy the orthogonality relations
\be\label{OR}
\sum_{\mu \in \OP(d)}2^{-\ell(\mu)- \delta(\sigma)}\frac{\zeta_\mu^\rho \zeta_\mu^ \sigma}{z_\mu}=\delta_{\rho,\sigma}
\ee
and 
\be\label{O2}
\sum_{\lambda \in \DP(d)} 2^{-\ell(\sigma) -\delta(\lambda)}\frac{\zeta_{\sigma}^\lambda \zeta_{\rho}^\lambda}{z_\sigma}=\delta_{\rho,\sigma}.
\ee
Let us also introduce the {\em central characters}
\be\label{cch}
f^\lambda_\mu =\frac{2^d d!}{2^{\ell(\mu)}z_\mu \dim V^\lambda} \zeta_\mu^\lambda.
\ee
Here
\be\label{dimV}
\dim V^\lambda=\zeta^\lambda_{1^d}=2^{\frac{\delta(\lambda)-\ell(\lambda)}{2}} d! Q_\lambda \Big|_{p_k=\delta_{k,1}}
\ee
is the dimension of the irreducible supermodule associated with the strict partition $\lambda$.

Let us consider the disconnected spin Hurwitz numbers for the ${\mathbb C}{\mathrm P}^1$ with the ramifications at $k$ branch points given by odd partitions $\mu_1,\dots,\mu_k$ with $|\mu_j|=d$.
The Gunningham formula \cite{Guni,LeeG} describes them in terms of the  central characters of the Sergeev group:
\be\label{spinc}
H_{d}^\theta (\mu_1,\dots,\mu_k)=2^{-d -\sum_{i=1}^k \ell^*(\mu_i)/2}\sum_{\lambda\in \DP(d)} 2^{-\delta(\lambda)} \left(\frac{\dim V^\lambda}{ d!}\right)^2 \prod_{j=1}^k f_{\mu_j}^\lambda,
\ee
where $\ell^*(\mu)=|\mu|-\ell(\mu)$ is the {\em colength} of the partition $\mu$.

\subsection{From spin Hurwitz numbers to 2-BKP hierarchy}
Let us single out two of the $k$ partitions and denote them by $\mu$ and $\nu$.
Using Equation (\ref{cch}) we can rewrite the spin Hurwitz numbers (\ref{spinc})  as follows
\be\label{spinc1}
H_d^\theta(\mu_1,\dots,\mu_{k-2},\mu,\nu)=2^{-\frac{1}{2}( \ell(\mu)+ \ell(\nu) +\sum_{i=1}^{k-2} \ell^*(\mu_i))}\sum_{\lambda\in \DP(d)} 2^{-\delta(\lambda)} \frac{\zeta_\mu^\lambda}{z_\mu} \frac{\zeta_\nu^\lambda}{z_\nu} \prod_{j=1}^{k-2} f_{\mu_j}^\lambda. 
\ee

Let us introduce $k-2$ families of weights $r^{(j)}_m$ for $1\leq j\leq k-2$, $m\in {\mathbb Z}_{+}$, associated with $k-2$ branch points. Then for the spin Hurwitz numbers (\ref{spinc1}) we introduce their weighted combinations
\be\label{ws}
H_{d,{\bf r}}^\theta(\nu,\mu)=\sum_{\mu_1,\dots, \mu_{k-2} \in \OP(d)} H_d^\theta(\mu_1,\dots,\mu_{k-2},\mu,\nu) R_{\mu_1}(r^{(1)})\dots R_{\mu_{k-2}}(r^{(k-2)}),
\ee
where 
\be\label{Rdef}
R_\mu(r)= \hbar^{-\ell^*(\mu)}  2^{-\ell(\mu)/2}\sum_{\sigma \in \DP(d)} \frac{\dim V^\sigma}{d! 2^{d/2}} 2^{-\delta(\sigma)}\zeta_\mu^\sigma r_\sigma.
\ee
Let us stress that the trivial ramifications $\mu_j=1^d$ are allowed in the summation. For the empty partition we put $R_\emptyset=1$.
If we compare this expression with the decomposition of the functions $p_\mu(Q_\sigma)$ in the basis of Schur Q-functions (\ref{pasQ}), then using Equation (\ref{dimV}) we get
\be\label{Rmu}
R_\mu(r)= \hbar^{-\ell^*(\mu)}  2^{-\ell^*(\mu)/2}p_\mu(Q_\sigma(\delta_{k,1})r_\sigma).
\ee
In particular,
\begin{equation}\label{rtoR}
\begin{split}
R_{[1]}(r)&=\frac{1}{2}Q_{[1]}(\delta_{k,1})r_1=r_1,\\
R_{[1,1]}(r)&=\frac{1}{2}Q_{[2]}(\delta_{k,1})r_2=r_2,\\
R_{[3]}(r)&={(2\hbar^2)}^{-1}\left(\frac{1}{2}Q_{[3]}(\delta_{k,1})r_3-\frac{1}{2}Q_{[2,1]}(\delta_{k,1})r_2r_1\right)=\frac{1}{3\hbar^2}(r_3-r_2r_1),\\
R_{[1,1,1]}(r)&=\frac{1}{2}Q_{[3]}(\delta_{k,1})r_3+\frac{1}{4}Q_{[2,1]}(\delta_{k,1})r_2r_1=\frac{1}{3}(2r_3+r_2r_1).
\end{split}
\end{equation}

Definition (\ref{Rdef}) is justified by the following observation: 
from the orthogonality relation (\ref{OR}) it follows that
\be
\sum_{\mu \in \OP(d)} 2^{-\ell^*(\mu)/2} f_{\mu}^{\lambda} R_\mu(r)=r_\lambda.
\ee
Therefore
\be\label{weight}
H_{d,{\bf r}}^\theta(\nu,\mu)=2^{-\frac{ \ell(\mu)+ \ell(\nu)}{2}} \sum_{\lambda\in \DP(d)} 2^{-\delta(\lambda)} \frac{\zeta_\mu^\lambda}{z_\mu} \frac{\zeta_\nu^\lambda}{z_\nu} \prod_{j=1}^{k-2} r^{(j)}_\lambda. 
\ee

By the Riemann--Hurwitz formula
\be
2-2g=\ell(\mu)+\ell(\nu)-\sum_{i=1}^{k-2} \ell^*(\mu_i),
\ee
where $g$ is the genus of the covering curve.
Consider the following generating function
\be\label{taugen}
\tau({\bf t},{\bf s})=\sum_{d=0}^\infty  \sum_{\mu,\nu\in \OP(d)} \hbar^{2g-2+\ell(\mu)+\ell(\nu)}  2^{-\frac{\ell(\mu)+\ell(\nu)}{2}} H_{d,{\bf r}}^\theta(\nu,\mu) \prod_{j=1}^{\ell(\mu)} \mu_j \prod_{j=1}^{\ell(\nu)} \nu_j t_\mu s_\nu.
\ee
Then for any choice of parameters $r_k^{(j)}$ from \eqref{weight} we have
\begin{theorem}The generating function $\tau({\bf t},{\bf s})$ is a hypergeometric tau-function of the 2-BKP hierarchy
\be\label{tauW}
\tau({\bf t},{\bf s})=\sum_{\lambda \in \DP}   \frac{Q_\lambda({\bf t}/2)Q_\lambda({\bf s}/2)}{2^{\ell(\mu)}}  \prod_{j=1}^{k-2} r^{(j)}_\lambda.
\ee
\end{theorem}
This tau-function can be identified with (\ref{hyperg}) if one puts
\be
T_m=\frac{1}{2}\sum_{j=1}^{k-2} \log r_m^{(j)}.
\ee

Similarly to the case of ordinary Hurwitz numbers, we can consider the limit when the maximal number of the branch points $k$ tends to infinity.


\subsection{Weighted spin Hurwitz numbers}

In the previous section we have constructed the weighted sums of the spin Hurwitz numbers that lead to the tau-functions of the 2-BKP hierarchy. While working with arbitrary parameters $r_k^{(j)}$ allow us to trace more information about the spin Hurwitz numbers from the properties of the tau-function, 
similarly to the case of the ordinary Hurwitz numbers \cite{Zog,GPH,ACEH3} we would like to introduce the distinguished weights, parametrized by one parameter $c_j$, $j=1,\dots, k-2$ for each of $k-2$ points.

By analogy with the ordinary weighted Hurwitz numbers, see \cite{GPH} and, more specifically,  in \cite[Equation (3.1)]{ACEH3}, one would tend to put $R_\mu(r^{(j)})=  c_j^{\ell^*(\mu)}$. For this choice the weighted spin Hurwitz numbers \eqref{ws} would be independent of $\hbar$. However, it is easy to see that $R_\mu$ should depend on $\hbar$ -- this is clear from (\ref{rtoR}).
Therefore, for the combinations of the spin Hurwitz numbers defined by (\ref{ws}) we need some ``completion'' of the partitions for the rational weight functions,
\be
R_\mu(r^{(j)})=\sum_{k=0}^{\infty} R_\mu^{(k)} \hbar^{2k},
\ee
a new effect of spin Hurwitz numbers which is absent in the theory of ordinary weighted Hurwitz numbers. 

Let us consider the generating function (\ref{tauW}) for the case with the maximal number of the branch points $k=3$. We claim that this tau-function can be considered as a generating function  of a spin version of dessins d'enfants. We also put $s_k=\delta_{k,1}\hbar^{-1}$, therefore the non-trivial branching is allowed only at two points. Then the tau-function \eqref{tauW} reduces to
\be
\tau({\bf t})= \sum_{\lambda \in \DP}   \frac{Q_\lambda({\bf t}/2)Q_\lambda(\delta_{k,1})}{2^{\ell(\mu)+|\lambda|} \, \hbar^{|\lambda|}}  r_\lambda,
\ee
where $r_\lambda = r^{(1)}_\lambda$. From the orthogonality relation \eqref{O2} and Equation (\ref{taugen}) it follows that 
\be\label{tauR}
\tau({\bf t})= \sum_{\mu \in \OP} \hbar^{-\ell (\mu)} 2^{-\frac{|\mu|+ \ell(\mu)}{2}}
\prod_{j=1}^{\ell(\mu) }\mu_j  t_\mu \frac{R_\mu(r)}{z_\mu}.
\ee

To relate $R_\mu$ to $\overline{\psi}(z)$ we use the results of Section \ref{S4} below.
From Equation \eqref{tauR} it follows that the coefficients $R_\mu(r)$'s are proportional to the coefficients of the correlation functions $W_n^\bullet$, namely
\be\label{WtoR}
W_n^\bullet= \hbar^{-n}  \sum_{\mu \in \OP, \ell(\mu)=n}   2^{-\frac{|\mu|+n}{2}}
R_\mu(r)  \sum_{\sigma \in S_n} X_{\sigma(1)}^{\mu_1}\dots X_{\sigma(n)}^{\mu_n}.
\ee
If we require $R_\mu^{(0)}=c^{\ell^*(\mu)}$, then for $n=1$ the leading term of equation (\ref{WtoR}) reduces to
\be
W_{0,1}(x)=\frac{1}{2} \sum_{k \in {{\mathbb Z}_{\odd}^+}}
\left(\frac{c}{\sqrt{2}}\right)^{k-1}X^k.
\ee
After identification of this expression with the general expression for the correlation function given by
Proposition \ref{prop:ClosedW01} we find $\psi(z)$,
\be
\psi(z)=\frac{1}{2}\log \frac{1+\sqrt{1+2c^2 z^2}}{2}.
\ee
One can identify it with $\overline{\psi}$, however, it is also possible to consider the $\hbar$-deformations.

Another possibility is to consider
\be
\overline{\psi}(z)=\frac{1}{2}\log\left(1+ \frac{c^2 z^2 }{2}  \right)
\ee
The associated functions $R_\mu$, (\ref{rtoR}) are rational functions of $c$, 
\begin{equation}
\begin{split}
R_{[1]}(r)&=1+ \frac{1}{8} c^2 \hbar^2 ,\\
R_{[1,1]}(r)&=\left(1+ \frac{1}{8} c^2 \hbar^2 \right)\left(1+ \frac{9}{8} c^2 \hbar^2\right),\\
R_{[3]}(r)&=c^2  \left(1+ \frac{1}{8} c^2 \hbar^2 \right)\left(1+ \frac{9}{8} c^2 \hbar^2\right),\\
R_{[1,1,1]}(r)&=\left(1+ \frac{1}{8} c^2 \hbar^2 \right)\left(1+ \frac{9}{8} c^2 \hbar^2\right)\left(1+ \frac{17}{8} c^2 \hbar^2\right).
\end{split}
\end{equation}




For $c=1$ the generating function $\tau({\bf t})$ for this choice of parametrization
 can be identified with the generalized BGW tau-function \cite{BGWP},
\be
\left.\tau_{BGW}({\bf t}/2)\right|_{ N^2 \mapsto -\frac{2}{\hbar^2}}= \tau({\bf t},{\delta_{k,1}\hbar^{-1}}).
\ee


\subsection{From BKP to KP} It is well known that the solutions of the BKP hierarchy are related to the solutions of the KP hierarchy for the particular choice of the variables \cite{JMBKP}. Following \cite{OBKP}, in this section we consider this relation for the hypergeometric tau-functions of both hierarchies. Namely, we relate any hypergeometric 2-BKP tau-function (\ref{hyperg}) to a hypergeometric tau-function of the 2-KP hierarchy.

Let us consider a 2-component system of neutral fermions $\phi_k^{(a)}$, $a=1,2$, satisfying the anti-commutation relations
\be
\left\{\phi_j^{(k)},\phi_l^{(m)}\right\}=(-1)^j\delta_{k,m}\delta_{j+l,0}.
\ee
Following \cite{JMroot} we can relate them to the {\em charged free fermions}
\be
\phi_j^{(1)}=\frac{\psi_j +(-1)^j \psi_{-j}^*}{\sqrt{2}},\,\,\,\,\,\,\,\,\, \phi_j^{(2)}=\sqrt{-1}\frac{\psi_j -(-1)^j \psi_{-j}^*}{\sqrt{2}}.
\ee
for $j\in {\mathbb Z}$. We immediately have
\be
 \psi_j \psi_j^* +\psi_{-j}^* \psi_{-j} =(-1)^j \left(  \phi_j^{(1)} \phi_{-j}^{(1)}+\phi_j^{(1)} \phi_{-j}^{(1)} \right).
\ee
Consider the bosonic operators for the charged fermions
\be
J_k^{\KP}=\sum_{j\in {\mathbb Z}}\Normord{ \psi_j \psi_{j+k}^*}.
\ee
Then for odd $k$ we have
\be
J_k^{\KP}=J_k^{(1)}+J_k^{(2)},
\ee
where the bosonic operators $J_k^{(j)}$ are given in terms of the corresponding neutral fermions by Equation (\ref{bosop}).

Let us consider the hypergeometric 2-KP tau-function
\be\label{KP2}
\tau_{\KP}({\bf t},{\bf s})=\lvac e^{J^{\KP}_+({\bf t})} e^{2\sum_{j=1}^\infty T_j (\Normord{  \psi_j \psi_j^*} -\Normord{ \psi_{-j}^* \psi_{-j}})}  e^{J^{\KP}_-({\bf s})} \rvac,
\ee
where $T_k$ are some parameters and $J_{\pm}^{\KP}({\bf t})=\sum_{k=1}^\infty t_k J_{\pm k}^{\KP}$.
Note that the group element in \eqref{KP2}  is not the most general diagonal group element. However, for this choice of the the group element we have a simple relation between this tau-function and a tau-function of the 2-BKP hierarchy. If all even time variables vanish, $t_{2k}=s_{2k}=0$ for $k\in {\mathbb Z}_+$, then
in terms of the neutral fermions we have
\be
\left.\tau_{\KP}({\bf t},{\bf s})\right|_{t_{2k}=s_{2k}=0}= \lvac e^{J^{(1)}_+({\bf t})+J^{(2)}_+({\bf t}) }  e^{2 \sum_{j=1}^\infty (-1)^j T_j (\Normord{  \phi^{(1)}_j  \phi^{(1)}_{-j}} +\Normord{\phi^{(2)}_j  \phi^{(2)}_{-j}} )}  e^{J^{(1)}_-({\bf s})+J^{(2)}_-({\bf s}) }  \rvac
\ee
Therefore \cite{OBKP}
\be\label{root}
\left.\tau_{\KP}({\bf t},{\bf s})\right|_{t_{2k}=s_{2k}=0}=\tau({\bf t},{\bf s})^2,
\ee
where
\be\label{Btau}
\tau({\bf t},{\bf s})= \lvac e^{J_+({\bf t}) }  e^{2 \sum_{j=1}^\infty (-1)^j T_j \Normord{  \phi_j  \phi_{-j}} }  e^{J_-({\bf s}) }  \rvac
\ee
is a hypergeometric tau-function of 2-BKP hierarchy \eqref{hyperg}. By definition, it depends only on odd times $t_{2k+1}$ and $s_{2k+1}$.

Let us compare the expansions of the tau-functions $\tau_{\KP}$ and $\tau$ in terms of the corresponding sets of the Schur functions. 
For the hypergeometric 2-KP tau-function (\ref{KP2}) one has
\be\label{2cKP}
\tau_{KP}({\bf t},{\bf s})=\sum_{\lambda} e^{2 \sum_{(i,j)\in \lambda}\overline{\psi}(\hbar(j-i-1/2))} s_\lambda({\bf t}) s_\lambda({\bf s}),
\ee
where $s_\lambda$ are the ordinary Schur functions and the sum runs over all partitions. These tau-functions are generating functions of the ordinary weighted Hurwitz numbers.
For the 2-component BKP tau-function (\ref{Btau}) we have
\be\label{hag}
\tau({\bf t},{\bf s})=\sum_{\lambda \in \DP}  e^{2\sum_{(i,j)\in \lambda}\overline{\psi}(\hbar (j-1/2))}  \frac{Q_\lambda({\bf t}/2) Q_\lambda({\bf s}/2)}{2^{\ell(\lambda)}}.
\ee


We see that for any hypergeometric tau-function of 2-BKP there exist a hypergeometric tau-function of 2-KP satisfying
(\ref{root}). It is easy to see that such 2-KP tau-function is not unique. Let $\lambda'$ denotes the transpose partition of $\lambda$. Then, as it follows i.e. from the Giambelli formula,
\be
\left.s_\lambda({\bf t})\right|_{t_{2k}=0}=\left.s_{\lambda'}({\bf t})\right|_{t_{2k}=0}
\ee
and for the Equation (\ref{2cKP}) we have
\begin{equation}
\begin{split}\label{dualsum}
\left.\tau_{KP}({\bf t},{\bf s})\right|_{t_{2k}=s_{2k}=0}&=\left.\sum_{\lambda} e^{2\sum_{(i,j)\in \lambda}\overline{\psi}(\hbar(j-i-1/2))} s_{\lambda'}({\bf t}) s_{\lambda'}({\bf s})\right|_{t_{2k}=s_{2k}=0}\\
&=\left.\sum_{\lambda} e^{2\sum_{(i,j)\in {\lambda'}}\overline{\psi}(\hbar(j-i-1/2))} s_{\lambda}({\bf t}) s_{\lambda}({\bf s})\right|_{t_{2k}=s_{2k}=0}\\
&=\left.\sum_{\lambda} e^{2\sum_{(i,j)\in {\lambda}}\overline{\psi}(\hbar(i-j-1/2))} s_{\lambda}({\bf t}) s_{\lambda}({\bf s})\right|_{t_{2k}=s_{2k}=0}.
\end{split}
\end{equation}
Therefore, (\ref{root}) is also satisfied for the tau-function (\ref{2cKP}) with $\overline{\psi}(z)$ substituted by $\overline{\psi}(-z-\hbar)$.

Hence, we can relate any generating function of the spin Hurwitz numbers \eqref{hag} to the generating function of the ordinary Hurwitz numbers \eqref{2cKP}. Moreover, we have at least two different hypergeometric tau-functions of the 2-KP hierarchy, corresponding to a given hypergeometric tau-function of the 2-BKP hierarchy. We expect that this identification should lead to a non-trivial relation between spin and ordinary Hurwitz numbers.

Let us consider a few examples.
If $T(x)=ax$, then $\overline{\psi}(z)=a$ is a constant, and the tau-function of the 2-BKP hierarchy is very simple
\begin{equation}
\begin{split}
\tau({\bf t},{\bf s})&=\sum_{\lambda \in \DP}  e^{2a|\lambda|}  \frac{Q_\lambda({\bf t}/2) Q_\lambda({\bf s}/2)}{2^{\ell(\lambda)}}\\
&=\exp\left(a \sum_{k\in {\mathbb Z}^{+}_{\odd}}k t_k s_k\right).
\end{split}
\end{equation}

More complicated example corresponds to $T(x)=\frac{bx^3}{3}+ax$ for some $a$ and $b$. It is associated with to $\overline{\psi}(z)=\frac{b}{2\hbar^2}z^2+\frac{b}{12}+a$.
The identity (\ref{root}) for this case with $a=\frac{2}{3}b$ and $\hbar=1$ was proven by Lee \cite{LeeRoot}. On the KP side the generating function, considered by Lee, is given by the last line of (\ref{dualsum}).


\section{Diagonal group element and $n$-point functions}\label{S4} 

In this section we  prove explicit closed algebraic formulas for the correlation functions $W_{g,n}$.


\subsection{Operators $\mathbb{J}_k$}

For the diagonal group element (\ref{Ddiag}) introduce the operators
\be
{\mathbb J}_k={\mathcal D}^{-1} J_k {\mathcal D},
\ee
Our first goal is to provide a few explicit formulas for these operators. To this end, we introduce a fermionic operator
\begin{equation}
\begin{split} \label{EE}
{\mathcal E} (u,a)& =\Normord{\phi(a^{-1} e^{u/2})\phi(-a^{-1} e^{-u/2})} \\
& =\sum_{k,m \in {\mathbb Z}} (-1)^m 
\Normord{\phi_{m-k} \phi_{-m}}
a^k e^{(m-k/2)u} \\
& =\sum_{k,m \in {\mathbb Z}} (-1)^m 
\hat{F}_{m-k,m} 
a^k e^{(m-k/2)u}. 
\end{split}
\end{equation}
Let
\be
{\mathcal S}(z)=\frac{e^{z/2}-e^{-z/2}}{z}.
\ee
Then in terms of the bosonic operators (\ref{bosop}) the operator ${\mathcal E} (u,a)$ is a reparametrization of the operator 
$\widehat{V}^{(2)}_B$ given by \eqref{Vr}, and
 can be represented as
\be \label{eq:DefinitionCalE}
{\mathcal E} (u,a) =\frac{1}{2}\frac{1+e^{-u}}{1-e^{-u}}\left(\exp\left(2u \sum_{k\in{\mathbb Z}_{\odd}^+}a^{-k} {\mathcal S}(ku) J_{-k}\right)\exp\left(2u\sum_{k\in{\mathbb Z}_{\odd}^+}a^k{\mathcal S}(ku) J_k\right) -1\right).
\ee

\begin{proposition}
The operators ${\mathbb J}_k$ belong to the image of the projective representation of $B_{\infty}$ for all $k \in {\mathbb Z}_{\odd}$
\be
{\mathbb J}_k= \frac{1}{2}\sum_{m\in {\mathbb Z}}(-1)^{m+1} e^{T_{k+m}-T_{-k-m}-T_m+T_{-m}} \hat{F}_{m,m+k}
\ee
and
\be
{\mathbb J}_k=\left.\frac{1}{2} [a^k] e^{2T(\p_u+1/2a\p_a)-2T(\p_u-1/2a\p_a)} {\mathcal E} (u,a)\right|_{u=0}.
\ee
\end{proposition}
\begin{proof}

From (\ref{comF}) we have
\be
\left[\sum_{a\in {\mathbb Z}}(-1)^{a} T_a F_{a,a},F_{k,m}\right]=(T_k-T_{-k}-T_m+T_{-m})F_{k,m}.
\ee
Hence
\be
{\mathcal D}^{-1}  \hat{F}_{k,m} {\mathcal D} = e^{-T_k+T_{-k}+T_m-T_{-m}}\hat{F}_{k,m}
\ee
and
\be
{\mathbb J}_k= \frac{1}{2}\sum_{m\in {\mathbb Z}}(-1)^{m+1} e^{T_{k+m}-T_{-k-m}-T_m+T_{-m}} \hat{F}_{m,m+k},
\ee
or, equivalently
\be
{\mathbb J}_k= \frac{1}{2}\sum_{m\in {\mathbb Z}}(-1)^{m+k+1} e^{T_{m}-T_{-m}-T_{m-k}+T_{k-m}} \hat{F}_{m-k,m}.
\ee
Comparing it to (\ref{EE}) we get
\be
{\mathbb J}_k=\left.\frac{1}{2} [a^k] e^{T(\p_u+1/2a\p_a) -T(-\p_u - 1/2a\p_a)-T(\p_u-1/2a\p_a)+T(-\p_u+1/2a\p_a)} {\mathcal E} (u,a)\right|_{u=0}.
\ee
\end{proof}


\subsection{Topological expansion}\label{S5.2}




 In terms of $\overline\psi$ we can rewrite the formula for ${\mathbb J}_k$, $k\in {\mathbb Z}_{\odd}^+$, as
\begin{align}
	{\mathbb J}_k & = \frac{1}{2}\sum_{m\in {\mathbb Z}}(-1)^{m+1} e^{T_{k+m}-T_{-k-m}-T_m+T_{-m}} \hat{F}_{m,m+k} 
	\\ \notag
	&
	= \frac{1}{2}\sum_{m\in {\mathbb Z}}(-1)^{m+1} \exp\left(2\sum_{i=1}^{k} \overline\psi(
	\hbar^2, \hbar (m-\frac 12 + i)) \right)  \hat{F}_{m,m+k} 
\end{align}

Let $\phi_k(y) \coloneqq \exp\left(2\sum_{i=1}^{k} \overline\psi(\hbar^2, y + \hbar (-\frac k2-\frac 12 + i)) \right) = \exp\left(2k \frac{{\mathcal S}(k\hbar \partial_y)}{\mathcal S(\hbar \partial_y)} \overline\psi(\hbar^2,y) \right)$. Then
\begin{align} \label{eq:JJ}
	{\mathbb J}_k & 
	= \frac{1}{2}\sum_{m\in {\mathbb Z}}(-1)^{m+1} \phi_k(\hbar(m+\frac k2)) \hat{F}_{m,m+k} 
	\\ \notag
	&
	=- \frac 12 \sum_{r=0}^\infty \partial_y^r  \phi_k(y) \Big|_{y=0} \cdot [u^r a^k]  {\mathcal E} (\hbar u,-a)
	\\ \notag
	&
	= -\frac 14  \sum_{r=0}^\infty \partial_y^r   \exp\left(2k \frac{{\mathcal S}(k\hbar \partial_y)}{\mathcal S(\hbar \partial_y)} \overline\psi(\hbar^2,y) \right) \Big|_{y=0}
	\\ \notag & \quad
	 [u^r a^k]  
	\frac{1+e^{-\hbar u}}{1-e^{-\hbar u}}\exp\left(2\hbar u \sum_{l\in{\mathbb Z}_{\odd}^+}(-a)^{-l}  {\mathcal S}(l\hbar u) J_{-l}\right)\exp\left(2\hbar u\sum_{l\in{\mathbb Z}_{\odd}^+}(-a)^l{\mathcal S}(l\hbar u) J_l\right) 
	\\ \notag
	&
	= \frac 14  \sum_{r=0}^\infty \partial_y^r \exp\left(2k \frac{{\mathcal S}(k\hbar \partial_y)}{\mathcal S(\hbar \partial_y)} \overline\psi(\hbar^2,y) \right) \Big|_{y=0}
	\\ \notag & \quad
	 [u^r a^k]  
	\frac{e^{\hbar u/2}+e^{-\hbar u/2}}{u\hbar \mathcal{S}(u\hbar)}\exp\left(2\hbar u \sum_{l\in{\mathbb Z}_{\odd}^+} a^{-l} {\mathcal S}(l\hbar u) J_{-l}\right)\exp\left(2\hbar u\sum_{l\in{\mathbb Z}_{\odd}^+} a^l{\mathcal S}(l\hbar u) J_l\right). 
\end{align} 

We also consider an arbitrary series $\overline y(\hbar^2, z) = \sum_{d=0}^\infty \hbar^{2d} y_d(z)$, where each $y_d(z)$ is an odd formal power series in $z$. The constant term of this series, $y_0=\overline y(0, z)$, is also denoted by $y=y(z)$. The prime object of our interest in this and the subsequent sections are the $\hbar$-expansions of the (disconnected) $n$-point functions
\be
H_n^\bullet=\left.\sum_{k_1,\dots,k_n \in {\mathbb Z}_{\odd}^+} \frac{\p^n \tau(\bf t,\hbar^{-1}{\bf s})}{\p t_{k_1}\dots \p t_{k_1}}\right|_{{\bf t}={\bf 0}}\frac{X_1^{k_1}}{k_1}\dots\frac{X_1^{k_n}}{k_n}
\ee
and
\be
W_n^\bullet =D_1\cdots D_n H_{n}^\bullet,
\ee
where $D_i\coloneqq X_i \partial_{X_i}$. Then from the definition of the operators ${\mathbb J}_{m}$ we have
\be
H_n^\bullet=\sum_{m_1,\dots,m_n\in {\mathbb Z}_{\odd}^+}^\infty \frac{X_1^{m_1}\dots X_n^{m_n}}{m_1\dots m_n} \lvac {\mathbb J}_{m_1}\dots \mathbb{ J}_{m_n} e^{\sum_{d=0}^\infty \hbar^{2d} \sum_{k\in{\mathbb Z}_{\odd}^+}\frac{J_{-k}}{\hbar k}[z^k] y_d(z)}\rvac
\ee
and
\begin{align} \label{eq:FirstFormulaDisconnectedW}
W_n^\bullet & = \sum_{m_1,\dots,m_n\in {\mathbb Z}_{\odd}^+} X_1^{m_1}\dots X_n^{m_n} \lvac \mathbb{J}_{m_1}\dots \mathbb{J}_{m_n} e^{\sum_{d=0}^\infty \hbar^{2d} \sum_{k\in{\mathbb Z}_{\odd}^+}\frac{J_{-k}}{\hbar k}[z^k] y_d(z)}\rvac.
\end{align}
Using the inclusion-exclusion formulas, we define the \emph{connected} $n$-point functions $H_n$ and $W_n$, $n\geq 1$, and they expand in $\hbar$ as 
\begin{align}
	H_n & = \sum_{g=0}^\infty H_{g,n} \hbar^{2g-2+n}; & W_n & = \sum_{g=0}^\infty W_{g,n} \hbar^{2g-2+n}.
\end{align}

\subsection{Preliminary formulas for $H_n^\bullet$ and $W_n^\bullet$} Denote
\be
B(z,w)\coloneqq \frac{zw}{(z-w)^2} + \frac {zw}{(z+w)^2}.
\ee 

\begin{proposition}\label{prop:ClosedFormula} We have the following formula for $H_{n}^\bullet$:
	\begin{align}\label{eq:ClosedFormula}
		H_{n}^\bullet = \sum_{m_1,\dots,m_n\in\mathbb{Z}_{\odd}^+} \prod_{i=1}^n \frac 1{m_i} X_i^{m_i}  \sum_{r_1,\dots,r_n=0}^\infty 
		\prod_{i=1}^n
		\partial_y^{r_i} \exp\left(2m_i \frac{{\mathcal S}(m_i\hbar \partial_y)}{\mathcal S(\hbar \partial_y)} \overline\psi \right) \Big|_{y=0}
		\\ \notag
		[\prod_{i=1}^n u_i^{r_i} z_i^{m_i}] \prod_{i=1}^n \frac{e^{\frac{\hbar u_i}2}+e^{-\frac{\hbar u_i}2}}{4 u_i\hbar \mathcal{S}(u_i\hbar)} e^{u_i \mathcal{S}(\hbar u_i z_i\partial_{z_i}) \overline y_i }
		\prod_{1\leq i<j\leq n} \left( e^{\hbar^2 u_iu_j \mathcal{S}(\hbar u_i z_i\partial_{z_i})\mathcal{S}(\hbar u_j z_j\partial_{z_j})B(z_i,z_j)} - 1\right),
	\end{align}
where $\overline\psi=\overline\psi(\hbar^2,y)$ and $\overline y_i=\overline y(\hbar^2,z_i)$. An analogous formula for $W_{n}^\bullet$ is obtained by replacing $\prod_{i=1}^n \frac 1{m_i} X_i^{m_i}$ in Equation~\eqref{eq:ClosedFormula} by $\prod_{i=1}^n X_i^{m_i}$.
\end{proposition}

\begin{proof} This formula should be understood as an expansion in the sector $|z_1|\ll |z_2|\ll\cdots \ll |z_n|\ll1$, and it comes directly from the commutation of the operators $\mathbb{J}_k$ given in Equation~\eqref{eq:JJ}, once one observes that 
\begin{equation}
\left[J_+(z), J_{-}(w)\right] = \frac 14 B(z,w).
\end{equation}
We refer also to an argument in~\cite[Section 3.2]{BDKS1}, which does exactly the same in a bit different situation. 
\end{proof}

\subsection{Closed algebraic formula for $W_{g,n}$}


We use a change of variables of exactly the type as in~\cite{ACEH1}, namely, 
\be
X=ze^{-2\psi(y(z))}.
\ee 
Define $D \coloneqq X\partial_X$ and define $Q$ by $D = Q^{-1} z\partial_z$. In the case we have variables $X_1,\dots,X_n$, we define $z_i$ by $X_i\coloneqq X(z_i)$, and furthermore we use the notation $D_i\coloneqq X_i\partial_{X_i}$, $Q_i\coloneqq z_i/X_i \cdot dX_i/dz_i$, $\overline y_i=\overline y(\hbar^2,z_i)$, $y_{i} = y(z_i)$, $\overline \psi_i = \psi(\hbar^2,y_{i})$, and $\psi_{i} = \psi(y_{i})$. 

\begin{theorem} \label{thm:ClosedFormulaForWgn} For $g\geq 0$, $n\geq 2$, $2g-2+n>0$, we have: 
	\begin{align} \label{eq:MainFormulaForWgn}
		& W_{g,n}  = [\hbar^{2g-2+n}]
		\sum_{\substack {j_1,\dots,j_n, \\ r_1,\dots,r_n = 0}}^\infty \left[\prod_{i=1}^n D_i^{j_i} [t_i^{j_i}]
		\frac 1{Q_i}
		e^{-2t_i \psi_{i} } 
		\partial_{y_i}^{r_i} e^{2t_i \frac{{\mathcal S}(t_i\hbar \partial_{y_{i}})}{\mathcal S(\hbar \partial_{y_{i}})} \overline\psi_i } [u_i^{r_i}]\right]
		\\ \notag &
		\prod_{i=1}^n \frac{e^{\frac{\hbar u_i}2}+e^{-\frac{\hbar u_i}2}}{4 u_i\hbar \mathcal{S}(u_i\hbar)} e^{u_i \left(\mathcal{S}(\hbar u_i z_i\partial_{z_i}) \overline y_i -y_{i} \right) }
		\sum_{\gamma\in\Gamma_n} 
		\prod_{(v_k,v_\ell)\in E_\gamma} \left( e^{\hbar^2 u_ku_\ell \mathcal{S}(\hbar u_k z_k\partial_{z_k})\mathcal{S}(\hbar u_\ell z_\ell\partial_{z_\ell})B(z_k,z_\ell)} - 1\right)
	\end{align}	
	Here $\Gamma_n$ is the set of all connected simple graphs on $n$ vertices $v_1,\dots,v_n$, and $E_\gamma$ is the set of edges of $\gamma$. 
\end{theorem}

\begin{remark} It is an explicit closed algebraic formula of the same type as in~\cite{BDKS1,BDKS2}. In particular the sum over ${j_1,\dots,j_n, r_1,\dots,r_n}$ is finite for every $(g,n)$. 
\end{remark}

\begin{remark} Note that $y_i$ and hence $\partial_{y_i}$ are odd in $z_i$. Note also that $\overline\psi_i$ and hence $\frac{{\mathcal S}(t_i\hbar \partial_{y_{i}})}{\mathcal S(\hbar \partial_{y_{i}})} \overline \psi_i$ are even in $z_i$. Note also that $D_i$ and $Q_i$ are even in $z_i$. Note also that in the second line the coefficient of $u_i^{r_i}$ for odd $r_i$ is even in $z_i$ and the coefficient of $u_i^{r_i}$ for even $r_i$ is odd in $z_i$. These observations imply that the right hand side of~\eqref{eq:MainFormulaForWgn} is necessarily odd in $z_1,\dots,z_n$. 
\end{remark}

\begin{remark}\label{rem:NoPolesDiag} Note that the structure of the formula suggests that there might be non-trivial poles along the diagonals $z_i=z_j$ and antidiagonals $z_i=-z_j$, but in fact the statement of the theorem in particular implies that these polar parts cancel and the resulting formula is non-singular at the diagonals and antidiagonals. Cf.~a discussion in~\cite[Remark 1.3 and Corollary 4.10]{BDKS1}.
\end{remark}

\begin{proof}[Proof of Theorem~\ref{thm:ClosedFormulaForWgn}]
	Recall the formula for $W_n^\bullet$ in Proposition~\ref{prop:ClosedFormula}. Passing to the connected $n$-point functions $W_n$ via inclusion-exclusion formula we replace 
	\begin{align}
		\prod_{1\leq i<j\leq n} \left( e^{\hbar^2 u_iu_j \mathcal{S}(\hbar u_i z_i\partial_{z_i})\mathcal{S}(\hbar u_j z_j\partial_{z_j})B(z_i,z_j)} - 1\right)
	\end{align}
	with 
	\begin{align}
		\sum_{\gamma\in\Gamma_n} 
		\prod_{(v_k,v_\ell)\in E_\gamma} \left( e^{\hbar^2 u_ku_\ell \mathcal{S}(\hbar u_k z_k\partial_{z_k})\mathcal{S}(\hbar u_\ell z_\ell\partial_{z_\ell})B(z_k,z_\ell)} - 1\right).
	\end{align}	
	Using the series expansion in $u_1,\dots,u_n$ (cf. \cite[Lemma 4.5]{BDKS1}), which is applicable only if $n\geq 2$ (hence the restriction on $n$ in the statement of the theorem) we can rewrite the formula as 
	\begin{align}\label{eq:ClosedFormulaWgnProof}
		W_{n} & = \sum_{m_1,\dots,m_n\in\mathbb{Z}_{\odd}^+} \prod_{i=1}^n  X_i^{m_i} \cdot  \sum_{r_1,\dots,r_n=0}^\infty 
		\prod_{i=1}^n
		\partial_y^{r_i} e^{2m_i \frac{{\mathcal S}(m_i\hbar \partial_y)}{\mathcal S(\hbar \partial_y)} \overline\psi_i }
		\\ \notag & \quad
		[\prod_{i=1}^n u_i^{r_i} z_i^{m_i}] \prod_{i=1}^n \frac{e^{\frac{\hbar u_i}2}+e^{-\frac{\hbar u_i}2}}{4 u_i\hbar \mathcal{S}(u_i\hbar)}  e^{u_i \left(\mathcal{S}(\hbar u_i z_i\partial_{z_i}) \overline y_i -y_{i} \right) }		\\ \notag & \quad
		\sum_{\gamma\in\Gamma_n} 
		\prod_{(v_k,v_\ell)\in E_\gamma} \left( e^{\hbar^2 u_ku_\ell \mathcal{S}(\hbar u_k z_k\partial_{z_k})\mathcal{S}(\hbar u_\ell z_\ell\partial_{z_\ell})B(z_k,z_\ell)} - 1\right).
	\end{align}
The next observation that we use is the following. Replace the summation in Equation~\eqref{eq:FirstFormulaDisconnectedW} from $\sum_{m_1,\dots,m_n\in\mathbb{Z}_{\odd}^+}$ to $\sum_{m_1,\dots,m_n\in\mathbb{Z}_{\odd}}$. This replacement changes the disconnected $W_n^\bullet$, but when we pass to the connected ones, this replacement just changes the $W_{0,2}$ by adding to it the singular term $\frac 14 B(X_1,X_2)$ (hence the condition $2g-2+n>0$ in the statement of the theorem), cf.~\cite[Proposition 4.1]{BDKS1}. With this adjustment, we have for $g\geq 0$, $n\geq 2$, $2g-2+n>0$: 
	\begin{align}\label{eq:ClosedFormulaWgnProof-2}
	W_{g,n} & = [\hbar^{2g-2+n}]\sum_{m_1,\dots,m_n\in\mathbb{Z}_{\odd}} \prod_{i=1}^n  X_i^{m_i} [z_i^{m_i}]\cdot  \sum_{r_1,\dots,r_n=0}^\infty 
	\prod_{i=1}^n
	\partial_y^{r_i} e^{2m_i \frac{{\mathcal S}(m_i\hbar \partial_y)}{\mathcal S(\hbar \partial_y)} \overline\psi_i }
	\\ \notag & \quad
	[\prod_{i=1}^n u_i^{r_i} ] \prod_{i=1}^n \frac{e^{\frac{\hbar u_i}2}+e^{-\frac{\hbar u_i}2}}{4 u_i\hbar \mathcal{S}(u_i\hbar)}  e^{u_i \left(\mathcal{S}(\hbar u_i z_i\partial_{z_i}) \overline y_i -y_{i} \right) }	\\ \notag & \quad
	\sum_{\gamma\in\Gamma_n} 
	\prod_{(v_k,v_\ell)\in E_\gamma} \left( e^{\hbar^2 u_ku_\ell \mathcal{S}(\hbar u_k z_k\partial_{z_k})\mathcal{S}(\hbar u_\ell z_\ell\partial_{z_\ell})B(z_k,z_\ell)} - 1\right).
\end{align}
Now note two things. First, $\partial_y^{r_i} \exp\left(2m_i \frac{{\mathcal S}(m_i\hbar \partial_y)}{\mathcal S(\hbar \partial_y)} \overline\psi_i \right)$ is even in $z_i$ for even $r_i$ and odd in $z_i$ for odd $r_i$. On the other hand, in the expression 		
\begin{align}
& [\prod_{i=1}^n u_i^{r_i} z_i^{m_i}] \prod_{i=1}^n \frac{e^{\frac{\hbar u_i}2}+e^{-\frac{\hbar u_i}2}}{4 u_i\hbar \mathcal{S}(u_i\hbar)} e^{u_i \left(\mathcal{S}(\hbar u_i z_i\partial_{z_i}) \overline y_i -y_{i} \right) }
\\ \notag &
\sum_{\gamma\in\Gamma_n} 
\prod_{(v_k,v_\ell)\in E_\gamma} \left( e^{\hbar^2 u_ku_\ell \mathcal{S}(\hbar u_k z_k\partial_{z_k})\mathcal{S}(\hbar u_\ell z_\ell\partial_{z_\ell})B(z_k,z_\ell)} - 1\right)
\end{align}
we only have terms with $r_i+m_i$ odd. 
This means that the whole expression in which we take the coefficients of $[\prod_{i=1}^n z_i^{m_i}]$ is odd in $z_1,\dots,z_n$, and thus we can extend the summation to ${m_1,\dots,m_n\in\mathbb{Z}}$. Second, we can use the trick that for a polynomial $g$ in $m$ we can replace $\sum X^m g(m)$ by $\sum_{j=o}^\infty D^j X^m [t^j] g(t)$. These two ideas allow to rewrite~\eqref{eq:ClosedFormulaWgnProof-2} as
	\begin{align}\label{eq:ClosedFormulaWgnProof-3}
	W_{g,n} & = [\hbar^{2g-2+n}]\sum_{j_1,\dots,j_n=0}^\infty D_i^{j_i} [t_i^{j_i}]\sum_{m_1,\dots,m_n\in\mathbb{Z}} \prod_{i=1}^n  X_i^{m_i} [z_i^{m_i}]e^{2m_i\psi_i}
	\\ \notag & \quad 
	 \sum_{r_1,\dots,r_n=0}^\infty 
	\prod_{i=1}^n e^{-2t_i\psi_i}
	\partial_y^{r_i} e^{2m_i \frac{{\mathcal S}(m_i\hbar \partial_y)}{\mathcal S(\hbar \partial_y)} \overline\psi_i }
	\\ \notag & \quad
	[\prod_{i=1}^n u_i^{r_i}] \prod_{i=1}^n \frac{e^{\frac{\hbar u_i}2}+e^{-\frac{\hbar u_i}2}}{4 u_i\hbar \mathcal{S}(u_i\hbar)}  e^{u_i \left(\mathcal{S}(\hbar u_i z_i\partial_{z_i}) \overline y_i -y_{i} \right) }	\\ \notag & \quad
	\sum_{\gamma\in\Gamma_n} 
	\prod_{(v_k,v_\ell)\in E_\gamma} \left( e^{\hbar^2 u_ku_\ell \mathcal{S}(\hbar u_k z_k\partial_{z_k})\mathcal{S}(\hbar u_\ell z_\ell\partial_{z_\ell})B(z_k,z_\ell)} - 1\right)
\end{align}
(here for each $r_1,\dots,r_n$ the second line expands in $\hbar$ with the coefficients that are manifestly polynomial in $t_1,\dots,t_n$). Finally, we apply the Lagrange--B\"uhrmann formula for $X_i = z_ie^{-2\psi_i}$ to Equation~\eqref{eq:ClosedFormulaWgnProof-3} (cf.~\cite[Lemma 4.7]{BDKS1}) and obtain the statement of the theorem.
\end{proof}

\subsection{Special cases} In this section we discuss the formulas for $W_{g,n}$ for $(g,n)=(0,2)$ and $n=1$, in the variable $z$ related to $X$ by $X=ze^{-2\psi(y(z))}$. 
We begin with unstable terms $(g,n)=(0,2)$ and $(0,1)$. 

\begin{proposition} \label{prop:W02} For $(g,n)=(0,2)$ we have:
	\be \label{eq:ClosedFormulaW02}
	W_{0,2} = \frac 1{4Q_1Q_2} B(z_1,z_2) - \frac 14 B(X_1,X_2).
	\ee
\end{proposition}

\begin{proof} Indeed, as we discussed in the proof of Theorem~\ref{thm:ClosedFormulaForWgn}, the change of summation from $m\in Z^+_{\odd}$ to $m\in Z_{\odd}$ and the commutation rules for $J_+(X_1), J_-(X_2)$ imply that the $(0,2)$ term gets a correction. We have:
	\begin{align}\label{eq:ProofClosedFormulaW02}
		W_{0,2}+\frac 14 B(X_1,X_2) & = 
		\sum_{m_1,m_2\in\mathbb{Z}_{\odd}} X_1^{m_1}X_2^{m_2} [z_1^{m_1}z_2^{m_2}]  
		e^{2m_1 \psi_1+2m_2 \psi_2}
		\frac 14 B(z_1,z_2)
		\\ \notag
		&  = 
		\sum_{m_1,m_2\in\mathbb{Z}} X_1^{m_1}X_2^{m_2} [z_1^{m_1}z_2^{m_2}]  
		e^{2m_1 \psi_1+2m_2 \psi_2}
		\frac 14 B(z_1,z_2)
		\\ \notag
		& = \frac 1{4Q_1Q_2} B(z_1,z_2).
	\end{align}
	In the second line, in order to change the summation from $m_1,m_2\in\mathbb{Z}_{\odd}$ to $m_1,m_2\in\mathbb{Z}$, we use that $B(z_1,z_2)$ is odd in $z_1$ and $z_2$  and $\psi_1 = \psi(y_1)$ (respectively, $\psi_2=\psi(y_2)$) is even in $z_1$ (respectively, $z_2$).
\end{proof}

\begin{proposition}\label{prop:ClosedW01}
	For $(g,n)=(0,1)$ we have: $W_{0,1}(X)= y(z)/2$.
\end{proposition}

\begin{proof} It is a straightforward computation. First, recall that
	\begin{align}\label{eq:01ClosedFormula}
		H_{0,1}(X) & = [\hbar^{-1}] \sum_{m\in\mathbb{Z}_{\odd}^+} \frac {X^{m}}{m}   \sum_{r=0}^\infty 
		\partial_y^r e^{2m \frac{{\mathcal S}(m\hbar \partial_y)}{\mathcal S(\hbar \partial_y)} \overline\psi(\hbar^2,y)} \Big|_{y=0}
		[u^{r} z^{m}] \frac{e^{\frac{\hbar u}2}+e^{-\frac{\hbar u}2}}{4 u\hbar \mathcal{S}(u\hbar)} e^{u\mathcal{S}(\hbar u z\partial_{z}) 
			\overline y(\hbar^2,z) }
		\\ \notag
		& = \sum_{m\in\mathbb{Z}_{\odd}^+} \frac {X^{m}}{m}   \sum_{r=0}^\infty 
		\partial_y^r e^{2m\psi(y)} \Big|_{y=0}
		[u^{r} z^{m}] \frac{e^{u y(z)}}{2 u } .
	\end{align}
	With this equation, in order to compute $W_{0,1}$, we consider its differential. We have:
	\begin{align}\label{eq:DD01ClosedFormula}
		DW_{0,1}(X) & = \sum_{m\in\mathbb{Z}_{\odd}^+} X^{m}  \sum_{r=0}^\infty 
		\partial_y^r e^{2m\psi(y)} \Big|_{y=0}
		[u^{r} z^{m}] z\partial_z \frac{e^{u y(z)}}{2 u } 
		\\ \notag & 
		= \sum_{m\in\mathbb{Z}_{\odd}^+} X^{m}  \sum_{r=0}^\infty 
		\partial_y^r e^{2m\psi(y)} \Big|_{y=0}
		[u^{r} z^{m}]  \frac{e^{u y(z)}QDy(z)}{2} 
		\\ \notag
		& = \sum_{m\in\mathbb{Z}_{\odd}^+} X^{m}[z^m] 
		e^{2m\psi(y(z)))}  \frac{QDy(z)}{2} 
		\\ \notag
		& = \sum_{m\in\mathbb{Z}} X^{m}[z^m] 
		e^{2m\psi(y(z)))}  \frac{QDy(z)}{2} 
		\\ \notag
		& = \frac 12 Dy(z).
	\end{align}
	Hence $W_{0,1}(X)= y(z)/2$.
\end{proof}

\subsubsection{Stable terms for $n=1$} Consider $g\geq 1$, $n=1$, that is, we consider $W_{1} = \sum_{g=0}^\infty \hbar^{2g-1} W_{g,1}(X)$. 

\begin{proposition} \label{prop:ClosedWg1} Under the change of variables $X=ze^{-2\psi(y(z))}$ we have:
	\begin{align} \label{eq:ClosedWg1}
		& W_{1}(X) = \frac{y}{2\hbar} + \sum_{j=1}^\infty D^{j-1} [t^j]
		e^{-2t\psi+2t \frac{{\mathcal S}(t\hbar \partial_y)}{\mathcal S(\hbar \partial_y)} \overline\psi} 
		\frac{Dy}{2\hbar} 
		\\ \notag 
		& + \sum_{j,r=0}^\infty D^j [t^j] \frac 1Q e^{-2t \psi}
		\partial_y^r e^{2t \frac{{\mathcal S}(t\hbar \partial_y)}{\mathcal S(\hbar \partial_y)} \overline\psi} 
		[u^r] 
		\left(\frac{e^{\frac{\hbar u}2}+e^{-\frac{\hbar u}2}}{4 u\hbar \mathcal{S}(u\hbar)}  e^{u \left(\mathcal{S}(\hbar u z\partial_{z}) \overline y -y\right) } \right).
	\end{align}
Here $y=y(z)$, $\overline y =\overline y(\hbar^2,z)$, $\psi=\psi(y)$, $\overline \psi = \overline \psi (\hbar^2,y)$, and $D=X\partial_X$. 
\end{proposition}

\begin{proof} By direct commutation of the operators, we have:
	\begin{align}\label{eq:g1ClosedFormula}
		 W_{1}(X) & = \sum_{m\in\mathbb{Z}_{\odd}^+} X^{m}  \sum_{r=0}^\infty 
		\partial_y^r e^{2m \frac{{\mathcal S}(m\hbar \partial_y)}{\mathcal S(\hbar \partial_y)} \overline\psi} \Big|_{y=0}
		[u^{r} z^{m}] \frac{e^{\frac{\hbar u}2}+e^{-\frac{\hbar u}2}}{4 u\hbar \mathcal{S}(u\hbar)} e^{u\mathcal{S}(\hbar u z\partial_{z}) 
			\overline y}
		\\ \notag
		&= \sum_{m\in\mathbb{Z}_{\odd}^+} X^{m}   \sum_{r=0}^\infty 
		\partial_y^r e^{2m \frac{{\mathcal S}(m\hbar \partial_y)}{\mathcal S(\hbar \partial_y)} \overline\psi} \Big|_{y=0}
		[u^{r} z^{m}] \left(\frac{e^{\frac{\hbar u}2}+e^{-\frac{\hbar u}2}}{4 u\hbar \mathcal{S}(u\hbar)} e^{u\mathcal{S}(\hbar u z\partial_{z}) 
			\overline y } - \frac{e^{u y}}{2 u \hbar } \right)
		\\ \notag
		& \quad +
		\sum_{m\in\mathbb{Z}_{\odd}^+} X^{m}  \sum_{r=0}^\infty 
		\partial_y^r e^{2m \frac{{\mathcal S}(m\hbar \partial_y)}{\mathcal S(\hbar \partial_y)} \overline\psi} \Big|_{y=0}
		[u^{r} z^{m}] \frac{e^{u y}}{2 u \hbar} 
	\end{align}
	The first summand is regular in $u$, so it can be computed as in the proof of Theorem~\ref{thm:ClosedFormulaForWgn}:
	\begin{align}\label{eq:g1ClosedFormula-1stsummand}
		&\sum_{m\in\mathbb{Z}_{\odd}^+} X^{m}   \sum_{r=0}^\infty 
		\partial_y^r e^{2m \frac{{\mathcal S}(m\hbar \partial_y)}{\mathcal S(\hbar \partial_y)} \overline\psi} \Big|_{y=0}
		[u^{r} z^{m}] \left(\frac{e^{\frac{\hbar u}2}+e^{-\frac{\hbar u}2}}{4 u\hbar \mathcal{S}(u\hbar)} e^{u\mathcal{S}(\hbar u z\partial_{z}) 
			\overline y } - \frac{e^{u y}}{2 u \hbar} \right)
		\\ \notag
		& =\sum_{m\in\mathbb{Z}} X^{m} [z^m]  \sum_{r=0}^\infty 
		\partial_y^r e^{2m \frac{{\mathcal S}(m\hbar \partial_y)}{\mathcal S(\hbar \partial_y)} \overline\psi} [u^r]
		\left(\frac{e^{\frac{\hbar u}2}+e^{-\frac{\hbar u}2}}{4 u\hbar \mathcal{S}(u\hbar)} e^{-uy+u\mathcal{S}(\hbar u z\partial_{z}) 
			\overline y } - \frac{1}{2 u \hbar } \right)
		\\ \notag
		& =\sum_{j,r=0}^\infty D^j [t^j] \frac 1Q e^{-2t \psi}
		\partial_y^r e^{2t \frac{{\mathcal S}(t\hbar \partial_y)}{\mathcal S(\hbar \partial_y)} \overline\psi} 
		[u^r]
		\left(\frac{e^{\frac{\hbar u}2}+e^{-\frac{\hbar u}2}}{4 u\hbar \mathcal{S}(u\hbar)} e^{-uy+u\mathcal{S}(\hbar u z\partial_{z}) 
			\overline y} \right).
	\end{align}
	The second summand of~\eqref{eq:g1ClosedFormula} can be computed by differentiation. 
	\begin{align}
		& D \sum_{m\in\mathbb{Z}_{\odd}^+} X^{m}  \sum_{r=0}^\infty 
		\partial_y^r e^{2m \frac{{\mathcal S}(m\hbar \partial_y)}{\mathcal S(\hbar \partial_y)} \overline\psi} \Big|_{y=0}
		[u^{r} z^{m}] \frac{e^{u y}}{2 u \hbar} 
		\\ \notag & 
		= \sum_{m\in\mathbb{Z}_{\odd}^+} X^{m}  \sum_{r=0}^\infty 
		\partial_y^r e^{2m \frac{{\mathcal S}(m\hbar \partial_y)}{\mathcal S(\hbar \partial_y)} \overline\psi} \Big|_{y=0}
		[u^{r} z^{m}] \frac{e^{u y(z)}QDy}{2\hbar} 
		\\ \notag 
		& = \sum_{m\in\mathbb{Z}_{\odd}^+} X^{m}  \sum_{r=0}^\infty 
		\partial_y^r e^{2m \frac{{\mathcal S}(m\hbar \partial_y)}{\mathcal S(\hbar \partial_y)} \overline\psi} 
		[u^{r} z^{m}] \frac{QDy}{2\hbar} 
		= \sum_{m\in\mathbb{Z}} X^{m}  [z^{m}]\sum_{r=0}^\infty 
		\partial_y^r e^{2m \frac{{\mathcal S}(m\hbar \partial_y)}{\mathcal S(\hbar \partial_y)} \overline\psi} 
		[u^{r}] \frac{QDy}{2\hbar} 
		\\ \notag 
		& = \sum_{j=0}^\infty D^j [t^j]
		e^{-2t\psi+2t \frac{{\mathcal S}(t\hbar \partial_y)}{\mathcal S(\hbar \partial_y)} \overline\psi} 
		\frac{Dy}{2\hbar} 
		= 	\frac{Dy}{2\hbar} + D \sum_{j=1}^\infty D^{j-1} [t^j]
		e^{-2t\psi+2t \frac{{\mathcal S}(t\hbar \partial_y)}{\mathcal S(\hbar \partial_y)} \overline\psi} 
		\frac{Dy}{2\hbar}. 
	\end{align}
	Combining these two computations, we obtain the statement of the proposition.	
\end{proof}


\section{Loop equations} \label{S5}
Consider the change of variables $X=ze^{-2\psi(y(z))}$. In this section we make a number of extra assumptions of analytical nature on the coefficients $\psi_{2d}$ and $y_{2d}$ of the $\hbar^2$-expansions of $\overline \psi$ and $\overline y$, and with these assumptions we prove the linear and quadratic loop equations for the $W_{g,n}$'s that we computed in the closed form in the previous section (or, more precisely, for the symmetric differentials that we construct from $W_{g,n}$'s).

\subsection{Assumptions}\label{sec:Assumptions} Let $z$ be a global affine coordinate on $\mathbb{C}\mathrm{P}^1$.
We assume that $\psi'(y(z))$ and $y'(z)$ can be analytically extended to rational functions on $\mathbb{C}\mathrm{P}^1$. These assumptions imply that $X=ze^{-2\psi(y(z))}$ extends to a global function on $\mathbb{C}\mathrm{P}^1$ and $d\log X$ is a rational $1$-form in the global coordinate $z$. 

The rational $1$-form $d\log X$ has a finite number of zeros, and we assume that all zeros of $d\log X$ are simple. We also assume that the coefficients of the positive degrees of $\hbar$ in the series $\overline \psi(\hbar^2,y(z))$ and $\overline y(\hbar^2,z)$ are rational functions in $z$ as well, and their singular points are disjoint from the zeros of $d\log X$.

\begin{proposition} Under these assumptions the symmetric $n$-differentials 
	\be \label{eq:omega-definition}
	\omega_{g,n}\coloneqq 2^{1-g}W_{g,n}(X_1,\dots,X_n) \prod_{i=1}^nd\log X_i+ \delta_{g,0}\delta_{n,2} \frac{1}{2} B(X_1,X_2) d\log X_1d\log X_2, \quad g\geq 0, n\geq 1,
	\ee
	analytically extend to global rational differentials on $(\mathbb{C}\mathrm{P}^1)^n$ for $2g-2+n>0$. 
\end{proposition}

\begin{proof} This statement follows directly from the structure of formulas given in Equations~\eqref{eq:MainFormulaForWgn} and~\eqref{eq:ClosedWg1}.
\end{proof}

Note the factor $2^{1-g}$. It is a compensation for the fact that the natural $W_{0,1}$ and $W_{0,2}$ that we obtained in the previous section are twice less than the formulas one might expect from the point of view of the spectral curve topological recursion, see Section~\ref{sec:TopoRec}.  

\subsection{Blobbed topological recursion}

Let $p$ be a simple zero point of $d\log X$. Let $\sigma$ denote the deck transformation of $X$ near $p$. 

\begin{definition} \label{def:Blobbed} We say that the system of symmetric $n$-differentials $\{\omega_{g,n}\}_{g\geq 0, n\geq 1}$ satisfies the linear loop equations at $p$ if for any $g\geq 0$, $n\geq 0$,
	\be
	\omega_{g,n+1}(w,z_{\llbracket n \rrbracket}) +
	\omega_{g,n+1}(\sigma(w),z_{\llbracket n \rrbracket})
	\ee
	is holomorphic at $ w\to p$ and vanishes at $w=p$.
	
	We say that the system of symmetric $n$-differentials $\{\omega_{g,n}\}_{g\geq 0, n\geq 1}$ satisfies the quadratic loop equations at $p$ if for any $g\geq 0$, $n\geq 0$, $(g,n)\not=(1,0)$, 
	\be
	\omega_{g-1,n+2}(w,\sigma(w),z_{\llbracket n \rrbracket}) +\sum_{\substack{g_1+g_2=g\\ I_1\sqcup I_2 = \llbracket n \rrbracket}}
	\omega_{g_1,n_1+1}(w,z_{I_2})\omega_{g_2,n_2+1}(\sigma(w),z_{I_2})
	\ee
	is holomorphic at $w\to p$ and has a zero of order at least two at $w=p$. In the case $(g,n)=(1,0)$ we require the same property, but we remove the singularity from the first summand, that is, we replace $\omega_{0,2}(w,\sigma(w))$ with 
	\begin{align}
	& \left(\omega_{0,2} -  \frac 12 B(X_1,X_2)d\log X_1d\log X_2\right)|_{X_1=X(w), X_2=X(\sigma(w))} 
	\\ \notag & = 2W_{0,2}(X_1,X_2) d\log X_1 d\log X_2|_{X_1=X(w), X_2=X(\sigma(w))}.		
	\end{align}
 
	If all zero points of $d\log X$ are simple and $\omega_{g,n}$'s satisfy the linear and quadratic loop equations at each of them, then we say that the system of symmetric $n$-differentials $\{\omega_{g,n}\}_{g\geq 0, n\geq 1}$ satisfies the blobbed topological recursion~\cite{BorotShadrin}.
\end{definition}

\begin{theorem}\label{thm:Blobbed} Under the analytic assumptions listed in Section~\ref{sec:Assumptions} the system of symmetric differentials~\eqref{eq:omega-definition} satisfies the blobbed topological recursion.
\end{theorem}

\subsection{Proof of Theorem~\ref{thm:Blobbed}}
Consider the connected correlation function defined as
	\begin{align}
	\mathcal{W}_{g,n}=\ & [\hbar^{2g-2+n}]\sum_{m_1,\dots,m_n\in{\mathbb Z}_{\odd}} 
	X_1^{m_1}\dots X_n^{m_n} 
	\\ \notag & \lvac \frac{J_{m_1}}{\hbar m_1} [a^0]\mathcal{E}(\hbar v, a) J_{m_2}\dots J_{m_n} {\mathcal D} e^{\sum_{d=0}^\infty \hbar^{2d} \sum_{k\in{\mathbb Z}_{\odd}^+}\frac{J_{-k}}{\hbar k}[z^k] y_d(z)} \rvac^\circ.
	\end{align}
(here by $\lvac - \rvac^\circ$ we mean the connected vacuum expectation obtained by inclusion-exclusion formula from the disconnected one).
	
	\begin{lemma} For $\sum_{m\in{\mathbb Z}_{\odd}} J_{m}X^m / (\hbar m) = \sum_{m\in{\mathbb Z}}J_{m}X^m / (\hbar m)$ we have:
		\begin{align}\label{eq:ConjugationE0}
			& \left[\sum_{m\in{\mathbb Z}}\frac{J_{m}X^m}{\hbar m},[a^0] \mathcal{E}(\hbar v, a)\right] = v{\mathcal S}(\hbar vX\partial_X) \left(\frac{\mathcal{E}(\hbar v, X) + \mathcal{E}(-\hbar v, X)  }{2}\right).
		\end{align}
	\end{lemma}
	\begin{proof} A straightforward computation using Equations~\eqref{eq:DefinitionCalE} and~\eqref{comJ}.
	\end{proof}
	\begin{remark} Note that $\mathcal{E}(u, X) = -\mathcal{E}(-u, -X)$. Hence the right hand side of Equation~\eqref{eq:ConjugationE0} is odd in $X$. Note also that the right hand side of Equation~\eqref{eq:ConjugationE0} is manifestly odd in $v$. 
	\end{remark}
	Our next goal is to compute the coefficients of $v^1$ and $v^3$ in \eqref{eq:ConjugationE0} applied to the covacuum. 
	\begin{lemma} We have:
		\begin{align}
			\lvac [v^1] v{\mathcal S}(\hbar vX\partial_X) \left(\frac{\mathcal{E}(\hbar v, X) + \mathcal{E}(-\hbar v, X)  }{2}\right) 
			& = 2 \lvac \sum_{m\in \mathbb{Z}^+_{\odd}} X^m J_m; \\ 
			\lvac [v^3] v{\mathcal S}(\hbar vX\partial_X) \left(\frac{\mathcal{E}(\hbar v, X) + \mathcal{E}(-\hbar v, X)  }{2}\right) 
			& = \hbar^2 \left(\frac 16(X\partial_X)^2 + \frac 1{6}\right) \lvac \sum_{m\in \mathbb{Z}^+_{\odd}} X^m J_m
			\\ \notag & \quad
			+\hbar^2 \frac 43 \lvac \Big(\sum_{m\in \mathbb{Z}^+_{\odd}} X^m J_m\Big)^3
		\end{align}	
	\end{lemma}
	\begin{proof} A straightforward computation using Equation~\eqref{eq:DefinitionCalE}.
	\end{proof}
	
	\begin{corollary} We have:
		\begin{align} \label{eq:ExpansionCurlyW}
			\mathcal{W}_{g,n}(X_{\llbracket n \rrbracket}) = {} & v\Bigg(2W_{g,n}(X_{\llbracket n \rrbracket}) \Bigg)  
		+ v^3
			\Bigg(
			\frac 43 \Big(
			W_{g-2,n+2}(X_1,X_1,X_1,X_{\llbracket n \rrbracket\setminus 1} ) 
			\\ \notag & 
			+ 3 \sum_{\substack{g_1+g_2=g-1 \\ I_1\sqcup I_2 = \llbracket n \rrbracket\setminus 1}} W_{g_1,n_1+1}(X_1,X_{I_1})W_{g_2,n_2+2}(X_1,X_1,X_{I_2})
			\\ \notag & 
			+ \sum_{\substack{g_1+g_2+g_3=g \\ I_1\sqcup I_2 \sqcup I_3= \llbracket n \rrbracket\setminus 1}} W_{g_1,n_1+1}(X_1,X_{I_1})W_{g_2,n_2+1}(X_1,X_{I_2})W_{g_3,n_3+1}(X_1,X_{I_3})
			\Big) 
			\\ \notag & 
			+  \left(\frac 16 (X_1\partial_{X_1})^2 + \frac 1{6}\right) W_{g-1,n}(X_{\llbracket n \rrbracket}) \Bigg)
			+ O(v^5),
		\end{align}
where we have to substitue $W_{0,2}(X_i,X_j)+\frac 14B(X_i,X_j)$ instead of $W_{0,2}(X_i,X_j)$ in all instances when the arguments are not the same, that is, $i\not= j$.
\end{corollary}
	
	Now, repeating \emph{mutatis mutandis} the arguments of the proofs of Theorem~\ref{thm:ClosedFormulaForWgn} and Propositions~\ref{prop:W02},~\ref{prop:ClosedW01}, and~\ref{prop:ClosedWg1}, we obtain closed algebraic formulas for $\mathcal{W}_{g,n}$.
	
	\begin{lemma} \label{lem:curlyw} Under the change of variables $X=ze^{-2\psi(y(z))}$ we have:
	\begin{align} \label{eq:MainFormulaForCurlyWgn}
		& \mathcal{W}_{g,n}(X_{\llbracket n \rrbracket}) = [\hbar^{2g-2+n}]
	\sum_{\substack {j_1,\dots,j_n, \\ r_1,\dots,r_n = 0}}^\infty \left[\prod_{i=2}^n D_i^{j_i} [t_i^{j_i}]
	\frac 1{Q_i}
	e^{-2t_i \psi_{i} } 
	\partial_{y_i}^{r_i} e^{2t_i \frac{{\mathcal S}(t_i\hbar \partial_{y_{i}})}{\mathcal S(\hbar \partial_{y_{i}})} \overline\psi_i } [u_i^{r_i}]\right]
	\\ \notag & 
	\left[D_1^{j_1} [t_1^{j_1}]
	\frac 1{Q_1}
	e^{-2t_1\psi_{1} } 
	\partial_{y_1}^{r_1} e^{2t_1 \frac{{\mathcal S}(t_1\hbar \partial_{y_{1}})}{\mathcal S(\hbar \partial_{y_{1}})} \overline\psi_1 }  v \mathcal{S}(t_1\hbar\partial_{y_1})  (e^{vy_1}+ e^{-vy_1}) [u_1^{r_1}]\right]
	\\ \notag &
	\prod_{i=1}^n \frac{e^{\frac{\hbar u_i}2}+e^{-\frac{\hbar u_i}2}}{4 u_i\hbar \mathcal{S}(u_i\hbar)} e^{u_i \left( \mathcal{S}(\hbar u_i z_i\partial_{z_i}) \overline y_i - y_{i} \right)}
	\sum_{\gamma\in\Gamma_n} 
	\prod_{(v_k,v_\ell)\in E_\gamma} \left( e^{\hbar^2 u_ku_\ell \mathcal{S}(\hbar u_k z_k\partial_{z_k})\mathcal{S}(\hbar u_\ell z_\ell\partial_{z_\ell})B(z_k,z_\ell)} - 1\right)
\end{align}	
for $n\geq 2$, $(g,n)\not=(0,2)$. Here $\Gamma_n$ is the set of all connected simple graphs on $n$ vertices $v_1,\dots,v_n$, and $E_\gamma$ is the set of edges of $\gamma$.

In the case $(g,n)=(0,2)$ we have
\be \label{eq:ClosedFormulaCurlW02}
\mathcal{W}_{0,2} = \frac {v(e^{vy_1}+e^{-vy_1})}{4Q_1Q_2} B(z_1,z_2) 
\ee

In the case $(g,n)=(0,1)$ we have
\be \label{eq:ClosedFormulaCurlW01}
\mathcal{W}_{0,1} = \frac 12 ( e^{vy_1}-e^{-vy_1}).
\ee

In the case $g\geq 1$, $n=1$ we have
\begin{align} \label{eq:ClosedFormulaCurlWg1}
 \mathcal{W}_{g,1} = {} & [\hbar^{2g-1}]
\sum_{\substack {j_1,r_1= 0}}^\infty 
D_1^{j_1} [t_1^{j_1}]
\frac 1{Q_1}
e^{-2t_1\psi_{1} } 
\partial_{y_1}^{r_1} e^{2t_1 \frac{{\mathcal S}(t_1\hbar \partial_{y_{1}})}{\mathcal S(\hbar \partial_{y_{1}})} \overline\psi_1 }  v \mathcal{S}(t_1\hbar\partial_{y_1})  (e^{vy_1}+ e^{-vy_1}) 
\\ \notag & \qquad \qquad \qquad
[u_1^{r_1}] \frac{e^{\frac{\hbar u_1}2}+e^{-\frac{\hbar u_1}2}}{4 u_1\hbar \mathcal{S}(u_1\hbar)} e^{-u_1 y_{1}+u_1 \mathcal{S}(\hbar u_1 z_1\partial_{z_1}) \overline y_1 }
\\ \notag & + [\hbar^{2g-1}]
\sum_{j_1=1}^\infty D^{j_1-1} [t^{j_1}]
e^{-2t_1\psi_{1} + 2t_1 \frac{{\mathcal S}(t_1\hbar \partial_{y_{1}})}{\mathcal S(\hbar \partial_{y_{1}})} \overline\psi_1 }  v \mathcal{S}(t_1\hbar\partial_{y_1})  (e^{vy_1}+ e^{-vy_1}) 
\frac{Dy}{2} .
\end{align}	
In all these formulas we use $y_i = y(z_i)$, $\overline y_i = \overline y(\hbar^2, z_i)$, $\psi_i = \psi(y_i)$, $\overline\psi_i = \overline\psi(\hbar^2,y_i)$, $Q_i = Q(z_i)$, $X_i=X(z_i)$, $D_i = X_i \partial_{X_i} = Q_i^{-1} z_i \partial_{z_i}$.  
\end{lemma}

\begin{proof} First, observe that $[a^0]\mathcal{E}(\hbar v, a)$ can be represented as the coefficient of $[\epsilon^1]$ in the expression
\begin{align} \label{eq:ExtralDIagonal}
	e^{\epsilon [a^0]\mathcal{E}(\hbar v, a)}=e^{\epsilon \sum_{k\in \mathbb{Z}} (-1)^k e^{\hbar vk} :\phi_k\phi_{-k}:} = e^{\epsilon \sum_{k\in \mathbb{Z}_+} (-1)^{k} (e^{\hbar vk}-e^{-\hbar vk}) :\phi_k\phi_{-k}:} 
\end{align}
For $T(k)=\frac 1{2\hbar} (e^{\hbar vk}-e^{-\hbar vk})$ we have $T(k+1)-T(k) = \Delta\overline\psi(\hbar^2,\hbar(k+\frac 12))$, where 
\begin{align}
\Delta\overline\psi(\hbar^2,y) = \frac 1{2\hbar} (e^{\frac{\hbar v}2} - e^{-\frac{hv}2})(e^{vy}+e^{-vy})= v +v^3\left(\frac{\hbar^2}{24} + \frac {y^2}2\right)+O(v^5).	
\end{align}
Define $\widetilde{ \mathbb{J}}_k$ as the conjugation of ${ \mathbb{J}}_k$ with the operator given in~\eqref{eq:ExtralDIagonal}. It is operator of exactly the same type as ${ \mathbb{J}}_k$, we just replace $\overline\psi(\hbar^2,y)$ by $\overline\psi(\hbar^2,y) + \epsilon \Delta\overline\psi(\hbar^2,y)$ in its definition. By~\eqref{eq:JJ} we have:
\begin{align}
\widetilde{ \mathbb{J}}_k	& = \frac 14  \sum_{r=0}^\infty \partial_y^r \exp\left(2k \frac{{\mathcal S}(k\hbar \partial_y)}{\mathcal S(\hbar \partial_y)} (\overline\psi(\hbar^2,y) + \epsilon \Delta\overline\psi(\hbar^2,y)) \right) \Big|_{y=0}
\\ \notag & \quad
[u^r a^k]  
\frac{e^{\hbar u/2}+e^{-\hbar u/2}}{u\hbar \mathcal{S}(u\hbar)}\exp\left(2\hbar u \sum_{l\in{\mathbb Z}_{\odd}^+} a^{-l} {\mathcal S}(l\hbar u) J_{-l}\right)\exp\left(2\hbar u\sum_{l\in{\mathbb Z}_{\odd}^+} a^l{\mathcal S}(l\hbar u) J_l\right).
\end{align}

Let  $\mathcal{W}_n = \sum_{g=0}^\infty \hbar^{2g-2+n} \mathcal{W}_{g,n}$. Then 
\begin{align}
	\mathcal{W}_{n}(X_{\llbracket n \rrbracket}) & =
	[\epsilon^1] \sum_{m_1,\dots,m_n\in {\mathbb Z}_{\odd}^+} \frac 1{m_1} X_1^{m_1}\dots X_n^{m_n}
		\\ \notag & \quad \lvac \widetilde{\mathbb{J}}_{m_1}\dots \mathbb{J}_{m_n} e^{\sum_{d=0}^\infty \hbar^{2d} \sum_{k\in{\mathbb Z}_{\odd}^+}\frac{J_{-k}}{\hbar k}[z^k] y_d(z)}\rvac^\circ.
\end{align}
By commutation of the operators, we have:
\begin{align}\label{eq:ClosedFormulaCurWDIsc}
	\mathcal{W}_{n} = {} & \sum_{m_1,\dots,m_n\in\mathbb{Z}_{\odd}^+} 
	\prod_{i=1}^n  X_i^{m_i} \cdot  \sum_{r_1,\dots,r_n=0}^\infty 
		\frac 1{m_1} [\epsilon^1] \partial_y^{r_1} \exp\left(2m_1 \frac{{\mathcal S}(m_1\hbar \partial_y)}{\mathcal S(\hbar \partial_y)} (\overline\psi + \epsilon \Delta\overline\psi(\hbar^2,y))\right) \Big|_{y=0}
			\\ \notag &
	\prod_{i=2}^n
	\partial_y^{r_i} \exp\left(2m_i \frac{{\mathcal S}(m_i\hbar \partial_y)}{\mathcal S(\hbar \partial_y)} \overline\psi \right) \Big|_{y=0}
	[\prod_{i=1}^n u_i^{r_i} z_i^{m_i}] \prod_{i=1}^n \frac{e^{\frac{\hbar u_i}2}+e^{-\frac{\hbar u_i}2}}{4 u_i\hbar \mathcal{S}(u_i\hbar)} e^{u_i \mathcal{S}(\hbar u_i z_i\partial_{z_i}) \overline y_i }
				\\ \notag &
	\sum_{\gamma\in\Gamma_n} 
\prod_{(v_k,v_\ell)\in E_\gamma} \left( e^{\hbar^2 u_ku_\ell \mathcal{S}(\hbar u_k z_k\partial_{z_k})\mathcal{S}(\hbar u_\ell z_\ell\partial_{z_\ell})B(z_k,z_\ell)} - 1\right),
\end{align}
where $\Gamma_n$ is the set of all connected simple graphs on $n$ vertices $v_1,\dots,v_n$, and $E_\gamma$ is the set of edges of $\gamma$.

As in the proof of Theorem~\ref{thm:ClosedFormulaForWgn} and Propositions~\ref{prop:ClosedW01} and~\ref{prop:ClosedWg1}, we use that for any formal power series $G(y)$ in $y$ and $F(u)$ in $u$, we have
\begin{align}
	\sum_{r=0}^\infty \partial_y^r G(y) |_{y=0} [u^r] F(u) = \sum_{r=0}^\infty \partial_y^r G(y)  [u^r] e^{-uy}F(u)
\end{align}
(\cite[Lemma 4.5]{BDKS1}). Applying it to Equation~\eqref{eq:ClosedFormulaCurWDIsc} in the cases $n\geq 2$, $(g,n)\not=(0,2)$ (these cases have to be treated separately, it is the same situation as in the proof of Theorem~\ref{thm:ClosedFormulaForWgn}), we obtain:
\begin{align}\label{eq:ClosedFormulaCurW-2}
	\mathcal{W}_{n} = {} & \sum_{m_1,\dots,m_n\in\mathbb{Z}_{\odd}^+} 
	\prod_{i=1}^n  X_i^{m_i} [z_i^{m_i}] e^{2m_i\psi_i} \cdot  \sum_{r_1,\dots,r_n=0}^\infty 
	\\ \notag &
	e^{-2m_1\psi_1}\partial_{y_1}^{r_1}\left( \exp\left(2m_1 \frac{{\mathcal S}(m_1\hbar \partial_{y_1})}{\mathcal S(\hbar \partial_{y_1})} \overline\psi_1 \right) v \mathcal{S}(m_1\hbar\partial_{y_1})  (e^{vy_1}+ e^{-vy_1})\right) 
	\\ \notag &
	\prod_{i=2}^n
	e^{-2m_i\psi_i}\partial_{y_i}^{r_i} \exp\left(2m_i \frac{{\mathcal S}(m_i\hbar \partial_{y_i})}{\mathcal S(\hbar \partial_{y_i})} \overline\psi_i \right) 
	[\prod_{i=1}^n u_i^{r_i}]\prod_{i=1}^n \frac{e^{\frac{\hbar u_i}2}+e^{-\frac{\hbar u_i}2}}{4 u_i\hbar \mathcal{S}(u_i\hbar)} e^{u_i \left(\mathcal{S}(\hbar u_i z_i\partial_{z_i}) \overline y_i -y \right)}
	\\ \notag &
	\sum_{\gamma\in\Gamma_n} 
	\prod_{(v_k,v_\ell)\in E_\gamma} \left( e^{\hbar^2 u_ku_\ell \mathcal{S}(\hbar u_k z_k\partial_{z_k})\mathcal{S}(\hbar u_\ell z_\ell\partial_{z_\ell})B(z_k,z_\ell)} - 1\right),
\end{align}
Starting from this point all further steps just repeat the computations made in the proof of Theorem~\ref{thm:ClosedFormulaForWgn}. We use three ideas: 
\begin{itemize}
	\item extend the summation to $m_1,\dots,m_n\in \mathbb{Z}$;
	\item capture the polynomial dependence on $m_1,\dots,m_n$ replacing their entrances by $t_1,\dots,t_n$ and applying $\prod_{i=1}^n\sum_{j_i=1}^\infty D_i^{j_i} [t_i^{j_i}]$;
	\item apply Lagrange--B\"uhrmann formula for the change of variables.
\end{itemize}
This completes the proof of Equation~\eqref{eq:MainFormulaForCurlyWgn}. All other equations stated in the lemma are obtained by small variations of this argument, which repeat the corresponding special cases in the proofs of Propositions~\ref{prop:W02},~\ref{prop:ClosedW01}, and~\ref{prop:ClosedWg1}.
\end{proof}

\begin{corollary}\label{cor:CurlyWLLoop} The functions $\mathcal{W}_{g,n}$, $g\geq 0$, $n\geq 1$, are formal power series in $v$, whose coefficients are rational functions in the variables $z_{\llbracket n \rrbracket}$, that near each simple zero point $p$ of $d\log X$ satisfy the property that
\be \label{eq:Property}
	\mathcal{W}_{g,n}(z_1,z_{\llbracket n \rrbracket\setminus 1}) + \mathcal{W}_{g,n}(\sigma(z_1),z_{\llbracket n \rrbracket\setminus 1})
\ee
is holomorphic at $z_1\to p$. Here $\sigma$ is the deck transformation of $X$ at $p$. 
\end{corollary}

\begin{proof} This follows directly from the structure of the formulas in Lemma~\ref{lem:curlyw}. We apply $D_1^{j_1}$ to a rational function that has a simple pole at $w\to p$ (coming from the factor $1/Q_1$). A function with at most simple pole automatically satisfies~\eqref{eq:Property}, and the operator $D_1$ preserves this property.
\end{proof}

\begin{proof}[Proof of Theorem~\ref{thm:Blobbed}] Fix a zero point $p$ of $d\log X$ (which by assumption is simple) and let $\sigma$ be the deck transformation of $X$ near this point. For any function $f(z)$ defined in the neighborhood of $p$ we define
\begin{equation}
	\mathsf{S}_zf(z) = f(z)+f(\sigma(z)).
\end{equation}
Then the linear loop equations at the point $p$ for the symmetric differentials expressed as in Equation~\eqref{eq:omega-definition} can be equivalently rewritten as
\begin{align}
	\mathsf{S}_{z_1} W_{g,n}(z_{\llbracket n \rrbracket})
\end{align}
is holomorphic at $z\to p$ for any $(g,n)$. Corollary~\ref{cor:CurlyWLLoop} applied to the coefficients of $[v^1]$ in $\mathcal{W}_{g,n}$ implies that it is indeed the case. 

Note also that Corollary~\ref{cor:CurlyWLLoop} applied to the coefficients of $[v^3]$ in $\mathcal{W}_{g,n}$ implies that $\mathsf{S}_{z_1}[v^3]\mathcal{W}_{g,n}$ is holomorphic at $z\to p$ for any $(g,n)$. Using explicit formula for $[v^3]2^{3-g}\mathcal{W}_{g,n}$ given in Equation~\eqref{eq:ExpansionCurlyW} and the linear loop equations, we conclude that
\begin{align} \label{eq:QLEPreliminary}
	 \mathsf{S}_{z_1}\Big( &
	2^{1-(g-2)}W_{g-2,n+2}(z_1,z_1,z_1,z_{\llbracket n \rrbracket\setminus 1} ) 
	\\ \notag & 
	+ 3 \sum_{\substack{g_1+g_2=g-1 \\ I_1\sqcup I_2 = \llbracket n \rrbracket\setminus 1}} 2^{1-g_1}W_{g_1,n_1+1}(z_1,z_{I_1})2^{1-g_2}W_{g_2,n_2+2}(z_1,z_1,z_{I_2})
	\\ \notag & 
	+ \sum_{\substack{g_1+g_2+g_3=g \\ I_1\sqcup I_2 \sqcup I_3= \llbracket n \rrbracket\setminus 1}} 2^{1-g_1}W_{g_1,n_1+1}(z_1,z_{I_1})2^{1-g_3}W_{g_2,n_2+1}(z_1,z_{I_2})2^{1-g_3}W_{g_3,n_3+1}(z_1,z_{I_3})
	\Big) 
\end{align}
is holomorphic at $z\to p$ for any $(g,n)$. Here we abuse the notation a little bit since each time we use $2W_{0,2}(z_i,z_j)$ with $i\not=j$, we actually mean $\frac 12B(z_i,z_j)$. 

This particular system of equations is studied in a bit different situation in~\cite[Lemma 20]{borot2017special}. The main difference between our situation and the one studied in~\cite[Lemma 20]{borot2017special} is the choice of $B$, which is the standard Bergman kernel in~\cite{borot2017special}, but it does not affect the proof in any step. Another difference is the rescaling of $W_{g,n}$ by $2^{1-g}$ in the definition of $\omega_{g,n}$'s, but both~\eqref{eq:QLEPreliminary} and the quadratic loop equations are homogeneous with respect to this rescaling. 

So, adjusted in our situation~\cite[Lemma 20]{borot2017special} proves that the holomorphy of the expression given in~\eqref{eq:QLEPreliminary} implies the quadratic loop equations for the symmetric differentials $\omega_{g,n}$ given by Equation~\eqref{eq:omega-definition}, under the condition that $y$ does not vanish at $z=p$. The latter condition is obviously satisfied in our situation. Indeed, the point $p$ satisfies the equation $1-2p\psi'(y(p))y'(p)=0$. On the other hand, $\psi'$ is an odd function in $y$, so at any point $z$ where $y(z)=0$, we have $1-2z\psi'(y(z))y'(z)=1$. Therefore, $y$ does not vanish at $z=p$. Hence the symmetric differentials $\omega_{g,n}$ satisfy the quadratic loop equations. 
\end{proof}


\section{Formulas for $H_{g,n}$}\label{S6}

In this section we derive expressions for $H_{g,n}$ by integration of the earlier derived expressions for $W_{g,n}$.
%
%
%
Since the case of $W_{g,1}$ was a bit special, we firstly perform a separate computation for $H_{g,1}$.

\begin{proposition}\label{prop:Hg1Formula} For $g\geq 1$ we have:
	\begin{align} \label{eq:Hg1formula}
		& H_{g,1} = 
		[\hbar^{2g}] \sum_{j=2}^\infty D^{j-2} [t^j]
		e^{-2t\psi+2t \frac{{\mathcal S}(t\hbar \partial_y)}{\mathcal S(\hbar \partial_y)} \overline\psi} 
		\frac{Dy}{2} 
		\\ \notag 
		& + [\hbar^{2g}]\sum_{j=1}^\infty D^{j-1} [t^j] \sum_{r=0}^\infty \frac 1Q e^{-2t \psi}
		\partial_y^r e^{2t \frac{{\mathcal S}(t\hbar \partial_y)}{\mathcal S(\hbar \partial_y)} \overline\psi} 
		[u^r] 
		\left(\frac{e^{\frac{\hbar u}2}+e^{-\frac{\hbar u}2}}{4 u\mathcal{S}(u\hbar)} e^{-uy+u\mathcal{S}(\hbar u z\partial_{z}) 
			\overline y } \right)
		\\ \notag
		& + [\hbar^{2g}] \int_0^z \left( \frac{1}{S(\hbar\partial_y)} \overline \psi - \psi\right) y'dz
		+ [\hbar^{2g}] \int_0^z\frac {dz}{2z} 
		\left(\overline y -y \right).
	\end{align}
	Here, as usual, we use $y= y(z)$, $\overline y = \overline y(\hbar^2, z)$, $\psi = \psi(y) = \psi(y(z))$, $\overline\psi= \overline\psi(\hbar^2,y)=\overline\psi(\hbar^2,y(z))$, $Q = Q(z)$, $X=X(z)$, $D = X\partial_{X} = Q^{-1} z \partial_{z}$.  
\end{proposition}

\begin{proof}
	Recall that for $g\geq 1$ 
	\begin{align}
		W_{g,1} =\ & [\hbar^{2g}] \sum_{j=1}^\infty D^{j-1} [t^j]
		e^{-2t\psi+2t \frac{{\mathcal S}(t\hbar \partial_y)}{\mathcal S(\hbar \partial_y)} \overline\psi} 
		\frac{Dy}{2} 
		\\ \notag 
		& + [\hbar^{2g}]\sum_{j,r=0}^\infty D^j [t^j] \frac 1Q e^{-2t \psi}
		\partial_y^r e^{2t \frac{{\mathcal S}(t\hbar \partial_y)}{\mathcal S(\hbar \partial_y)} \overline\psi} 
		[u^r] 
		\left(\frac{e^{\frac{\hbar u}2}+e^{-\frac{\hbar u}2}}{4 u\mathcal{S}(u\hbar)} e^{-uy+u\mathcal{S}(\hbar u z\partial_{z}) 
			\overline y } \right).
	\end{align}
	Hence,  for $g\geq 1$
	\begin{align}
		& H_{g,1} = [\hbar^{2g}] \sum_{j=2}^\infty D^{j-2} [t^j]
		e^{-2t\psi+2t \frac{{\mathcal S}(t\hbar \partial_y)}{\mathcal S(\hbar \partial_y)} \overline\psi} 
		\frac{Dy}{2} 
		\\ \notag 
		& + [\hbar^{2g}]\sum_{j=1}^\infty D^{j-1} [t^j] \sum_{r=0}^\infty \frac 1Q e^{-2t \psi}
		\partial_y^r e^{2t \frac{{\mathcal S}(t\hbar \partial_y)}{\mathcal S(\hbar \partial_y)} \overline\psi} 
		[u^r] 
		\left(\frac{e^{\frac{\hbar u}2}+e^{-\frac{\hbar u}2}}{4 u\mathcal{S}(u\hbar)} e^{-uy+u\mathcal{S}(\hbar u z\partial_{z}) 
			\overline y } \right)
		\\ \notag
		& + [\hbar^{2g}] \int_0^z dz \frac Qz [t^1]
		e^{-2t\psi+2t \frac{{\mathcal S}(t\hbar \partial_y)}{\mathcal S(\hbar \partial_y)} \overline\psi} 
		\frac{Dy}{2} 
		\\ \notag
		& + [\hbar^{2g}] \int_0^z dz \frac Qz [t^0] \sum_{r=0}^\infty \frac 1Q e^{-2t \psi}
		\partial_y^r e^{2t \frac{{\mathcal S}(t\hbar \partial_y)}{\mathcal S(\hbar \partial_y)} \overline\psi} 
		[u^r] 
		\left(\frac{e^{\frac{\hbar u}2}+e^{-\frac{\hbar u}2}}{4 u\mathcal{S}(u\hbar)} e^{-uy+u\mathcal{S}(\hbar u z\partial_{z}) 
			\overline y } \right)
	\end{align}
	(note that the constant term in $z$ of this expression vanishes).
	The third term here can be computed as
	\begin{align}
		[\hbar^{2g}] \int_0^z dz \frac Qz [t^1]
		e^{-2t\psi+2t \frac{{\mathcal S}(t\hbar \partial_y)}{\mathcal S(\hbar \partial_y)} \overline\psi} 
		\frac{Dy}{2} 
		& = [\hbar^{2g}] \int_0^z dz \left( \frac{1}{S(\hbar\partial_y)} \overline \psi  - \psi\right) \frac{QDy}{z}
		\\ \notag
		& = [\hbar^{2g}] \int_0^z \left( \frac{1}{S(\hbar\partial_y)} \overline \psi - \psi\right) y'dz.
	\end{align}
	The fourth term can be computed as
	\begin{align}
		& [\hbar^{2g}] \int_0^z dz \frac Qz [t^0] \sum_{r=0}^\infty \frac 1Q e^{-2t \psi}
		\partial_y^r e^{2t \frac{{\mathcal S}(t\hbar \partial_y)}{\mathcal S(\hbar \partial_y)} \overline\psi} 
		[u^r]
		\left(\frac{e^{\frac{\hbar u}2}+e^{-\frac{\hbar u}2}}{4 u\mathcal{S}(u\hbar)} e^{-uy+u\mathcal{S}(\hbar u z\partial_{z}) 
			\overline y } \right)
		\\ \notag 
		& = [\hbar^{2g}] \int_0^z\frac {dz}z 
		[u^0] 
		\left(\frac{e^{\frac{\hbar u}2}+e^{-\frac{\hbar u}2}}{4 u\mathcal{S}(u\hbar)} e^{-uy+u\mathcal{S}(\hbar u z\partial_{z}) 
			\overline y } \right)
		\\ \notag 
		& = [\hbar^{2g}] \int_0^z\frac {dz}{2z} 
		\left(\overline y -y \right) = [\hbar^{2g}] \int_0^z\frac {dz}{2z} 
		\left(\overline y -y \right).
	\end{align}
	Combining these formulas, we obtain the statement of the proposition.
\end{proof}

In the case $n=2$ we have the following formula for $H_{g,2}$.

\begin{proposition} \label{prop:Hg2Formula} In the case $g=0$ we have
	\begin{align}
		H_{0,2} = \frac 14 \log \frac{(z_1-z_2)(X_1+X_2)}{(z_1+z_2)(X_1-X_2)}.
	\end{align}
	For $g\geq 0$ we have:
	\begin{align} \label{eq:Hg2Formula}
		H_{g,2} = \ & [\hbar^{2g}] 
		\sum_{\substack {j_1,j_2 =1 \\ r_1,r_2 = 0}}^\infty \Bigg[\prod_{i=1}^2 D_i^{j_i-1} [t_i^{j_i}]
		\frac 1{Q_i}
		e^{-2t_i \psi_i } 
		\partial_{y_i}^{r_i} e^{2t_i \frac{{\mathcal S}(t_i\hbar \partial_{y_i})}{\mathcal S(\hbar \partial_{y_i})} \overline \psi_i } [u_i^{r_i}]
		\\ \notag &
		\frac{e^{\hbar u_i/2}+e^{-\hbar u_i/2}}{4 u_i\hbar \mathcal{S}(u_i\hbar)} e^{-u_iy_i+ u_i  \mathcal{S}(\hbar u_i z_i\partial_{z_i}) \overline y_i }\Bigg]
		\left( e^{\hbar^2 u_1u_2 \mathcal{S}(\hbar u_1 z_1\partial_{z_1})\mathcal{S}(\hbar u_2 z_2\partial_{z_2})B(z_1,z_2)} - 1\right)
		\\ \notag
		& + [\hbar^{2g}] 
		\sum_{\substack {j =1 \\ r = 0}}^\infty \Bigg[ D_1^{j-1} [t^{j}]
		\frac 1{Q_1}
		e^{-2t \psi_1 } 
		\partial_{y_1}^{r} e^{2t\frac{{\mathcal S}(t\hbar \partial_{y_1})}{\mathcal S(\hbar \partial_{y_1})} \overline \psi_1 } [u^{r}]
		\\ \notag &
		\frac{e^{\hbar u/2}+e^{-\hbar u/2}}{4 u\hbar \mathcal{S}(u\hbar)} e^{-uy_1+ u \mathcal{S}(\hbar u z_1\partial_{z_1}) \overline y_1 }\Bigg]
		\frac{1}{2}\hbar u \mathcal{S}(\hbar u z_1\partial_{z_1}) \Big( \frac{z_1}{z_1-z_2} - \frac{z_1}{z_1+z_2} \Big)
		\\ \notag
		& + [\hbar^{2g}] 
		\sum_{\substack {j =1 \\ r = 0}}^\infty \Bigg[ D_2^{j-1} [t^{j}]
		\frac 1{Q_2}
		e^{-2t \psi_2 } 
		\partial_{y_2}^{r} e^{2t\frac{{\mathcal S}(t\hbar \partial_{y_2})}{\mathcal S(\hbar \partial_{y_2})} \overline \psi_2 } [u^{r}] 
		\\ \notag &
		\frac{e^{\hbar u/2}+e^{-\hbar u/2}}{4 u\hbar \mathcal{S}(u\hbar)} e^{-uy_2+ u \mathcal{S}(\hbar u z_2\partial_{z_2}) \overline y_2 }\Bigg]
		\frac{1}{2}\hbar u \mathcal{S}(\hbar u z_2\partial_{z_2}) \Big( \frac{z_2}{z_2-z_1} - \frac{z_2}{z_2+z_1} \Big).
	\end{align}
	Here we use the notation $y_i = y(z_i)$, $\overline y_i = \overline y(\hbar^2, z_i)$, $\psi_i = \psi(y_i)$, $\overline\psi_i = \overline\psi(\hbar^2,y_i)$, $Q_i = Q(z_i)$, $X_i=X(z_i)$, $D_i = X_i \partial_{X_i} = Q_i^{-1} z_i \partial_{z_i}$ for $i=1,2$.
\end{proposition}

\begin{proof} Note that all formulas above are odd in both arguments, hence they vanish if any of their arguments vanishes. Hence it is enough to check that $D_1D_2H_{g,2} = W_{g,2}$. In the case $g=0$ it is a straightforward computation. For $g\geq 1$ we recall the relevant special case of Equation~\eqref{eq:MainFormulaForWgn}:
	\begin{align} 
		W_{g,2}  =\  & [\hbar^{2g}] 
		\sum_{\substack {j_1,j_2 \\ r_1,r_2 = 0}}^\infty \left[\prod_{i=1}^2 D_i^{j_i} [t_i^{j_i}]
		\frac 1{Q_i}
		e^{-2t_i \psi_i } 
		\partial_{y_i}^{r_i} e^{2t_i \frac{{\mathcal S}(t_i\hbar \partial_{y_i})}{\mathcal S(\hbar \partial_{y_i})} \overline \psi_i } [u_i^{r_i}]\right]
		\\ \notag &
		\prod_{i=1}^2 \frac{e^{\hbar u_i/2}+e^{-\hbar u_i/2}}{4 u_i\hbar \mathcal{S}(u_i\hbar)} e^{-u_iy_i+ u_i  \mathcal{S}(\hbar u_i z_i\partial_{z_i}) \overline y_i }
		\left( e^{\hbar^2 u_1u_2 \mathcal{S}(\hbar u_1 z_1\partial_{z_1})\mathcal{S}(\hbar u_2 z_2\partial_{z_2})B(z_1,z_2)} - 1\right).
	\end{align}	
	Note that
	\begin{align}
		\sum_{j,r=0}^\infty D^{j} [t^{j}]
		\frac 1{Q}
		e^{-2t\psi } 
		\partial_{y}^{r} e^{2t \frac{{\mathcal S}(t\hbar \partial_{y})}{\mathcal S(\hbar \partial_{y})} \overline \psi } [u^r]
		= 	\frac 1Q[u^0] + \sum_{r=0}^\infty\sum_{j=1}^\infty D^{j} [t^{j}]
		\frac 1{Q}
		e^{-2t\psi } 
		\partial_{y}^{r} e^{2t \frac{{\mathcal S}(t\hbar \partial_{y})}{\mathcal S(\hbar \partial_{y})} \overline \psi }[u^r]
	\end{align}
	The second summand here can be trivially integrated by applying $D^{-1}$. In order to integrate the cases of application of $\frac 1Q[u^0]$ we observe that
	for any $\gamma\in\Gamma_n$ the coefficient of $u_i^0$ in
	\begin{align}
		& \frac 1{Q_i}[u_i^0] \frac{e^{\hbar u_1/2}+e^{-\hbar u_i/2}}{4 u_i\hbar \mathcal{S}(u_i\hbar)} e^{-u_iy_i+ u_i  \mathcal{S}(\hbar u_i z_i\partial_{z_i}) \overline y_i }
		\left( e^{\hbar^2 u_iu_k \mathcal{S}(\hbar u_i z_i\partial_{z_i})\mathcal{S}(\hbar u_k z_k\partial_{z_k})B(z_i,z_k)} - 1\right)
		\\ \notag 
		& =  D_i \frac{1}{2}\hbar u_k \mathcal{S}(\hbar u_k z_k\partial_{z_k}) \Big( \frac{z_k}{z_k-z_i} - \frac{z_k}{z_k+z_i} \Big).
	\end{align}
	(here $k=2$ if $i=1$ and $k=1$ if $i=2$), which also admits application of $D_i^{-1}$. In particular, if we apply this term for both variables, we have:
	\begin{align} 
		& \prod_{i=1}^2 \left[\frac 1{Q_i} [u_i^0] \frac{e^{\hbar u_i/2}+e^{-\hbar u_i/2}}{4 u_i\hbar \mathcal{S}(u_i\hbar)} e^{-u_iy_i+ u_i  \mathcal{S}(\hbar u_i z_i\partial_{z_i}) \overline y_i }\right]
		\\ \notag &
		\left( e^{\hbar^2 u_1u_2 \mathcal{S}(\hbar u_1 z_1\partial_{z_1})\mathcal{S}(\hbar u_2 z_2\partial_{z_2})B(z_1,z_2)} - 1\right)
		\\ \notag
		& =\frac 1{4Q_1Q_2}B(z_1,z_2),
	\end{align}	 
	so this case doesn't contribute to $[\hbar^{2g}]$, $g\geq 1$. Combining these computations with the application of $D_1^{-1}D_2^{-1}$, we obtain the statement of the proposition. 
\end{proof}

Finally, in the general case of $n\geq 3$ we have the following expression for $H_{g,n}$.

\begin{proposition} \label{prop:HgnFormula} For a $\gamma\in\Gamma_n$ let $I_\gamma$ denote the subset of vertices of $\gamma$ of index $\geq 2$. Let $K_\gamma\subset E_\gamma$ be the subset of the set of edges that connect a vertex of index $1$ to another vertex. When we write $(v_i,v_k)\in K_\gamma$, we assume that $v_i$ is the vertex of index $1$ (and, therefore, $v_k\in I_\gamma$). We have:
	\begin{align} \label{eq:HgnGeneralFormula}
		H_{g,n} =\ & [\hbar^{2g-2+n}]\sum_{\gamma\in\Gamma_n}\prod_{i\in I_\gamma} \left[ \sum_{r_i=0}^\infty\sum_{j_i =1}^\infty D_i^{j_i-1} [t_i^{j_i}]
		\frac 1{Q_i}
		e^{-2t_i \psi_i } 
		\partial_{y_i}^{r_i} e^{2t_i \frac{{\mathcal S}(t_i\hbar \partial_{y_i})}{\mathcal S(\hbar \partial_{y_i})} \overline \psi_i } [u_i^{r_i}] \right]
		\\
		\notag
		& 
		\prod_{i\in I_\gamma} \frac{e^{\hbar u_i/2}+e^{-\hbar u_i/2}}{4 u_i\hbar \mathcal{S}(u_i\hbar)} e^{-u_i y_i + u_i ( \mathcal{S}(\hbar u_i z_i\partial_{z_i}) \overline y_i }
		\\ 
		\notag
		&
		\prod_{(v_k,v_\ell)\in E_\gamma\setminus K_\gamma} \left( e^{\hbar^2 u_ku_\ell \mathcal{S}(\hbar u_k z_k\partial_{z_k})\mathcal{S}(\hbar u_\ell z_\ell\partial_{z_\ell})B(z_k,z_\ell)} - 1\right)
		\\
		\notag
		& 
		\prod_{(v_i,v_k)\in  K_\gamma} \Bigg( 
		\frac{1}{2}\hbar u_k \mathcal{S}(\hbar u_k z_k\partial_{z_k}) \Big( \frac{z_k}{z_k-z_i} - \frac{z_k}{z_k+z_i} \Big) +
		\\ 
		\notag
		& 
		\left[ \sum_{r_i=0}^\infty\sum_{j_i =1}^\infty D_i^{j_i-1} [t_i^{j_i}]
		\frac 1{Q_i}
		e^{-2t_i \psi_i } 
		\partial_{y_i}^{r_i} e^{2t_i \frac{{\mathcal S}(t_i\hbar \partial_{y_i})}{\mathcal S(\hbar \partial_{y_i})} \overline \psi_i } [u_i^{r_i}]\right]
		\\
		\notag 
		& \frac{e^{\hbar u_i/2}+e^{-\hbar u_i/2}}{4 u_i\hbar \mathcal{S}(u_i\hbar)} e^{-u_i y_i + u_i ( \mathcal{S}(\hbar u_i z_i\partial_{z_i}) \overline y_i }
		\left( e^{\hbar^2 u_iu_k \mathcal{S}(\hbar u_iz_i\partial_{z_i})\mathcal{S}(\hbar u_k z_k\partial_{z_k})B(z_i,z_k)} - 1\right) \Bigg) . 
	\end{align}
	Here, as usual, we use $y_i = y(z_i)$, $\overline y_i = \overline y(\hbar^2, z_i)$, $\psi_i = \psi(y_i)$, $\overline\psi_i = \overline\psi(\hbar^2,y_i)$, $Q_i = Q(z_i)$, $X_i=X(z_i)$, $D_i = X_i \partial_{X_i} = Q_i^{-1} z_i \partial_{z_i}$.  
\end{proposition}

\begin{proof} Note that $H_{g,n}$ as given in Equation~\eqref{eq:HgnGeneralFormula} vanishes if we set any of its variables to zero (since it is odd in each of its variables). So, the only thing that we have to check is that indeed $D_1\cdots D_n H_{g,n} = W_{g,n}$ as given by Equation~\eqref{eq:MainFormulaForWgn}. Recall Equation~\eqref{eq:MainFormulaForWgn}:
	\begin{align} 
		W_{g,n}  = & [\hbar^{2g-2+n}] 
		\sum_{\substack {j_1,\dots,j_n, \\ r_1,\dots,r_n = 0}}^\infty \left[\prod_{i=1}^n D_i^{j_i} [t_i^{j_i}]
		\frac 1{Q_i}
		e^{-2t_i \psi_i } 
		\partial_{y_i}^{r_i} e^{2t_i \frac{{\mathcal S}(t_i\hbar \partial_{y_i})}{\mathcal S(\hbar \partial_{y_i})} \overline \psi_i } [u_i^{r_i}] \right]
		\\ \notag &
		\prod_{i=1}^n \frac{e^{\hbar u_i/2}+e^{-\hbar u_i/2}}{4 u_i\hbar \mathcal{S}(u_i\hbar)} e^{-u_iy_i+ u_i  \mathcal{S}(\hbar u_i z_i\partial_{z_i}) \overline y_i }
		\\ \notag &
		\sum_{\gamma\in\Gamma_n} 
		\prod_{(v_k,v_\ell)\in E_\gamma} \left( e^{\hbar^2 u_ku_\ell \mathcal{S}(\hbar u_k z_k\partial_{z_k})\mathcal{S}(\hbar u_\ell z_\ell\partial_{z_\ell})B(z_k,z_\ell)} - 1\right).
	\end{align}	
	Note that
	\begin{align}
		\sum_{j,r=0}^\infty D^{j} [t^{j}]
		\frac 1{Q}
		e^{-2t\psi } 
		\partial_{y}^{r} e^{2t \frac{{\mathcal S}(t\hbar \partial_{y})}{\mathcal S(\hbar \partial_{y})} \overline \psi } [u^r]
		= 	\frac 1Q[u^0] + \sum_{r=0}^\infty\sum_{j=1}^\infty D^{j} [t^{j}]
		\frac 1{Q}
		e^{-2t\psi } 
		\partial_{y}^{r} e^{2t \frac{{\mathcal S}(t\hbar \partial_{y})}{\mathcal S(\hbar \partial_{y})} \overline \psi }[u^r].
	\end{align}
	The second summand here can be trivially integrated by applying $D^{-1}$. In order to integrate the cases of application of $\frac 1Q[u^0]$ we observe that
	for any $\gamma\in\Gamma_n$ the coefficient of $u_i^0$ in 
	\begin{align}
		& \frac{e^{\hbar u_i/2}+e^{-\hbar u_i/2}}{4 u_i\hbar \mathcal{S}(u_i\hbar)} e^{-u_iy_i+ u_i  \mathcal{S}(\hbar u_i z_i\partial_{z_i}) \overline y_i }
		\prod_{(v_i,v_k)\in E_\gamma} \left( e^{\hbar^2 u_iu_k \mathcal{S}(\hbar u_i z_i\partial_{z_i})\mathcal{S}(\hbar u_k z_k\partial_{z_k})B(z_i,z_k)} - 1\right)
	\end{align}
	is non-trivial if and only if $i$ has index $1$ in $\gamma$. Then there is only one edge $(e_i,e_k)\in E_\gamma$ that is attached to the vertex $i$. In this case, 
	\begin{align}
		& \frac 1{Q_i} [u_i^0] \frac{e^{\hbar u_i/2}+e^{-\hbar u_i/2}}{4 u_i\hbar \mathcal{S}(u_i\hbar)} e^{-u_iy_i+ u_i  \mathcal{S}(\hbar u_i z_i\partial_{z_i}) \overline y_i }
		\left( e^{\hbar^2 u_iu_k \mathcal{S}(\hbar u_i z_i\partial_{z_i})\mathcal{S}(\hbar u_k z_k\partial_{z_k})B(z_i,z_k)} - 1\right)
		\\ \notag
		& = \frac 1{Q_i} \frac 12 \hbar u_k\mathcal{S}(\hbar u_k z_k\partial_{z_k})B(z_i,z_k)
		= D_i \frac{1}{2}\hbar u_k \mathcal{S}(\hbar u_k z_k\partial_{z_k}) \Big( \frac{z_k}{z_k-z_i} - \frac{z_k}{z_k+z_i} \Big),
	\end{align}
	and we can apply $D_i^{-1}$ to the latter expression. This explains the special summands for $(v_i,v_k)\in K_\gamma$ in Equation~\eqref{eq:HgnGeneralFormula} and completes the proof of the proposition.
\end{proof}


\section{Topological recursion for spin Hurwitz number with completed cycles}\label{S7}

The goal of this Section is to prove a conjecture proposed by Giacchetto, Kramer, and Lewa\'nski. In our terms, it concerns the symmetric $n$-differentials constructed from Orlov's hypergeometric 2-BKP tau-functions for $\overline \psi=\frac 12 {\mathcal S}(\hbar \partial_y) y^{2s}$ and $\overline y = z$. But in fact we consider a more general situation, with $\overline \psi=\frac 12 {\mathcal S}(\hbar \partial_y) P(y)$ and $\overline y =y =R(z)$, where $P$ is an arbitrary even polynomial in $y$ and $R$ is an arbitrary odd polynomial in $z$, since the arguments in this more general situation do not differ from the ones for the Giacchetto--Kramer--Lewa\'nski situation. 

\begin{remark}
Note that if we put  $\overline{\psi}(y)=\frac{1}{2} {\mathcal S}(\hbar \p_y) P(y)$, then the weight for the KP hypergeometric tau-function (\ref{2cKP})  does not coincide with the deformation, considered in \cite{BDKS2}. Therefore, if in the relation (\ref{root}) one of the tau-functions, $\tau_{KP}$ or $\tau$, is described by a suitable version of topological recursion, the other one is not described by it.
\end{remark}

\subsection{Topological recursion in the odd situation} \label{sec:TopoRec}

Consider $\mathbb{C}\mathrm{P}^1$ with a fixed global coordinate $z$, and with two functions, $X$ and $y$ such that $X(-z)=-X(z)$ and $y(-z)=-y(-z)$, with an extra assumption that $dX/X$ is a rational differential with the simple critical points $p_1,\dots,p_N$ (it is clear that $N$ must be even and the set of critical points is invariant under $z\leftrightarrow -z$) and $y$ is holomorphic near the critical points with $dy|_{p_i}\not=0$. It is not necessary but both convenient and sufficient for our goals to assume that $y$ is meromorphic. Let 
\begin{align}
	\mathcal{B}(z_1,z_2)\coloneqq \frac 12 \left(\frac{1}{(z_1-z_2)^2} + \frac{1}{(z_1+z_2)^2}\right)dz_1dz_2.
\end{align}

With this input we construct a system of symmetric differentials $\omega_{g,n}$, $g\geq 0$, $n\geq 1$, given by 
\begin{align}
	& \omega_{0,1}(z_1) = y(z_1) d\log X(z_1) ; 
	\\ \notag
	& \omega_{0,2}(z_1,z_2) = \mathcal{B}(z_1,z_2);
\end{align}
and for $2g-2+n>0$ we use the recursion
\begin{align}\label{eq:TopologicalRecursion}
	\omega_{g,n} (z_1,\dots,z_n) \coloneqq \ & \frac 12 \sum_{i=1}^N \res_{z\to p_i} 
	\frac{\int_z^{\sigma_i(z)} \mathcal{B}(z_1,\cdot)} {\omega_{0,1}(\sigma_i(z_1)) -  \omega_{0,1}(z_1) }\Bigg(
	\omega_{g-1,n+1}(z,\sigma_i(z),z_{\llbracket n \rrbracket \setminus \{1\}})
	\\
	\notag
	& + \sum_{\substack{g_1+g_2 = g, I_1\sqcup I_2 = {\llbracket n \rrbracket \setminus \{1\}} \\
			(g_1,|I_1|),(g_2,|I_2|) \not= (0,0) }} \omega_{g_1,1+|I_1|}(z,z_{I_1})\omega_{g_2,1+|I_2|}(\sigma_i(z), z_{I_2})\Bigg),
\end{align}
where  $\sigma_i$ is the deck transformation of $X$ near $p_i$, $i=1,\dots, N$. We wouldn't go into the discussion of this peculiar version of this topological recursion, as it should be done in a more general equivariant setup. 

For our goals it is sufficient to state the following equivalent reformulation of this version of topological recursion, which is completely parallel to~\cite[Theorem 2.2]{BorotShadrin} and~\cite[Section 1]{BorotEynardOrantin}.

\begin{lemma} \label{lem:EquiavelentTR} A system of meromorphic symmetric differentials $\omega_{g,n}$,
	 $2g-2+n>0$ is obtained from the given starting data (that includes the formulas for $\omega_{0,1}$ 
	 and $\omega_{0,2}$)
	 by topological recursion~\eqref{eq:TopologicalRecursion} if and only if 
	\begin{enumerate}
		\item This system of differentials satisfies the blobbed topological recursion (see Definition~\ref{def:Blobbed}).
		\item For any $g\geq 0,n\geq 1,2g-2+n>0$ 
		\begin{align}
			\omega_{g,n}(z_{\llbracket n \rrbracket}) = \sum_{i_1,\dots,i_n=1}^N \Bigg(\prod_{j=1}^n \res_{w_j\to p_{i_j}} \int_{p_{i_j}}^{w_i} \mathcal{B}(\cdot,z_j)\Bigg)\omega_{g,n}(w_{\llbracket n \rrbracket})
		\end{align}
		(this is the so-called projection property).
	\end{enumerate}
\end{lemma}

\begin{proof} The same argument as in~\cite[Section 2.4]{BorotShadrin}. 
\end{proof}

If we represent the symmetric differential $\omega_{g,n}$ as $\omega_{g,n} = 2^{1-g}W_{g,n} \prod_{i=1}^n d\log X_i$, $\omega_{0,2}= 2W_{0,2}d\log X_1d\log X_2 + \mathcal{B}(X_1,X_2) $, where $W_{g,n} = D_1\cdots D_n H_{g,n}$, $D_i = X_i\partial_{X_i}$, then the linear loop equations in combination with the projection property can be equivalently reformulated in terms of $H_{g,n}$. This reformulation can be directly applied in the odd case that we consider here, and we recall it and prove for the particular $n$-point functions of spin Hurwitz numbers with completed cycles in the next section, Section~\ref{sec:QuasiPol}. 

\subsection{Quasi-polynomiality}\label{sec:QuasiPol} The goal of this section is to prove some special property of the functions $H_{g,n}$, and, as a corollary, $W_{g,n}$'s that is sometimes called quasi-polynomiality in the literature and in the context of topological recursion is equivalent to a combination of the so-called projection property and the liner loop equations. We refer to~\cite[Section 3]{BDKS2} for a full discussion.

Recall that with $\overline \psi=\frac 12 {\mathcal S}(\hbar \partial_y) P(y)$, $P(-y)=P(y)$ is a polynomial, and $\overline y = y = R(z)$, $R(-z)=-R(z)$ is a polynomial we have $X=z\exp(-P(R(z)))$. Let $p_1,\dots,p_N\in\mathbb{C}\mathrm{P}^1$ be the critical points of $X$. Here $N=\deg P\cdot \deg R\in 2\mathbb{Z}$, we assume that all critical points are simple, and the set of critical points is obviously invariant under the involution $z\leftrightarrow -z$.

Define the space $\Theta_n$ as the linear span of functions $\prod_{i=1}^N f_i(z_i)$, where each $f_i(z_i)$ is a rational function on $\mathbb{C}\mathrm{P}^1$, $f_i(-z_i)=-f_i(z_i)$, $f_i$ has poles only at the points $p_1,\dots,p_N$, and the principal part of $f_i$ at $p_k$, $k=1,\dots,N$, is odd with respect to the corresponding deck transformation $\sigma_k$ of function $X$ near $p_k$. The last condition can be reformulated as a requirement that for any $k=1,\dots,N$ the locally defined function $f_i(z_i)+f_i(\sigma_k z_i)$ is holomorphic at $z_i\to p_k$. 

\begin{proposition}\label{prop:QuasiPol} In the case  $\overline \psi=\frac 12 {\mathcal S}(\hbar \partial_y) P(y)$, $P(-y)=P(y)$ is a polynomial, and $\overline y = y = R(z)$, $R(-z)=-R(z)$ is a polynomial, the functions $H_{g,n}$ belong to the space $\Theta_n$, for any $n\geq 1$, $g\geq 0$ such that $2g-2+n>0$. 
\end{proposition}

\begin{proof} In the proof we analyze the formulas obtained in Propositions~\ref{prop:Hg1Formula},~\ref{prop:Hg2Formula}, and~\ref{prop:HgnFormula}. It is clear from the structure of the formulas~\eqref{eq:Hg1formula},~\eqref{eq:Hg2Formula}, and \eqref{eq:HgnGeneralFormula} that with our assumptions $H_{g,n}$ are rational functions in $z_1,\dots,z_n$. 
	
	Consider $H_{g,n}$ as a function of $z_1$, treating the rest of the variables as parameters. From the shape of the formula we see that it might have poles at $z_1 \to \pm z_i$, $i=2,\dots,n$, $z_1\to \infty$, and at the zeros of $Q$. In this case $Q = z\partial_z\log X = 1+ z\partial_z P(R(z))$, and its zeros are exactly $p_1,\dots, p_N$. 
	
	From Remark~\ref{rem:NoPolesDiag} it follows that there are no singularities at $z_1=\pm z_i$, $i=2,\dots,n$. 
	
	In all terms of the formulas~\eqref{eq:Hg1formula},~\eqref{eq:Hg2Formula}, and \eqref{eq:HgnGeneralFormula} the principal part at $z_1\to p_k$ is generated by the iterative application of the operator $D_1=X_1\partial_{X_1} = Q(z_1)^{-1} z_1\partial_{z_1}$ to a function that is either holomorphic at $z_1\to p_k$ (as in the first summand of~\eqref{eq:Hg1formula}), or has a simple pole at $z_1\to p_k$  (as in the second summand of~\eqref{eq:Hg1formula}, where we divide a function holomorphic at $z_1\to p_k$ by $Q(z_1)$). Holomorphic functions and functions with a simple pole automatically have principal parts at $z_1\to p_k$ that are odd with respect to the deck transformation at $p_k$, and the operator $D_1$ preserves this property (while increasing the order of the pole at $p_k$). 
	
	Let us now check that there is no pole at $z_1\to \infty$. Note that the terms that really look special, the last two summands in Equation~\eqref{eq:Hg1formula}, vanish with our assumptions (and that is crucially important since for any other choice of $\overline \psi$ and $\overline y$ with given $\psi=P$ and $y=R$ it wouldn't be the case). To all other terms in the formulas~\eqref{eq:Hg1formula},~\eqref{eq:Hg2Formula}, and \eqref{eq:HgnGeneralFormula} the same rough estimation of the order of pole is applicable, cf.~\cite[Lemma 4.6]{BDKS2}. We perform it here only for the second summand in Equation~\eqref{eq:Hg1formula}, since in all other cases the analysis is exactly the same.
	To this end, consider
	\begin{align}\label{eq:OrderPole1}
		& [\hbar^{2g}]\sum_{j=1}^\infty D^{j-1} [t^j] \sum_{r=0}^\infty \frac 1Q e^{-2t \psi}
		\partial_y^r e^{2t \frac{{\mathcal S}(t\hbar \partial_y)}{\mathcal S(\hbar \partial_y)} \overline\psi} 
		[u^r] 
		\left(\frac{e^{\frac{\hbar u}2}+e^{-\frac{\hbar u}2}}{4 u\mathcal{S}(u\hbar)} e^{-uy+u\mathcal{S}(\hbar u z\partial_{z}) 
			\overline y } \right)
		\\ \notag
		&  =	[\hbar^{2g}]\sum_{j=1}^\infty D^{j-1} [t^j] \sum_{r=0}^\infty \frac 1Q 
		(\partial_y+2t P'(y))^r e^{2t ({{\mathcal S}(t\hbar \partial_y)} -1)P(y)} 
		\\ \notag & \qquad \qquad \qquad \quad \quad \quad \quad
		 [u^r] 
		\left(\frac{e^{\frac{\hbar u}2}+e^{-\frac{\hbar u}2}}{4 u\mathcal{S}(u\hbar)} e^{u(\mathcal{S}(\hbar u z\partial_{z})-1) 
			R(z) } \right)
	\end{align}
	Note that the operator $D=Q^{-1} z\partial z$ decreases the order of pole at $z\to \infty$ by $\deg Q = \deg P\deg R$. The same holds for the factor $Q^{-1}$ alone. This means that the order of pole in~\eqref{eq:OrderPole1} at $z\to \infty$ is equal to the order of pole at $z\to \infty$ of 
	\begin{align}\label{eq:OrderPole2}
		&	[\hbar^{2g}] \sum_{r=0}^\infty 
		(\partial_y+2t P'(y))^r e^{2t ({{\mathcal S}(t\hbar \partial_y)} -1)P(y)} \Big|'_{t = z^{-\deg P\deg R}}
		\\ \notag & \qquad \qquad 
	[u^r]
		\left(\frac{e^{\frac{\hbar u}2}+e^{-\frac{\hbar u}2}}{4 u\mathcal{S}(u\hbar)} e^{u(\mathcal{S}(\hbar u z\partial_{z})-1) 
			R(z) } \right),
	\end{align}
	where by $|'$ we mean that we only select the terms with $\deg t \geq 1$. With this substitution observe that each application of the operator $\partial_y+2t P'(y)$ decreases the order at $z\to \infty$ by $\deg R$. Thus the order of pole of \eqref{eq:OrderPole2} at $z\to \infty$ is equal to the order of pole at $z\to \infty$ of
	\begin{align}\label{eq:OrderPole3}
		&	[\hbar^{2g}] 
		e^{2t ({{\mathcal S}(t\hbar \partial_y)} -1)P(y)} 
		\left(\frac{e^{\frac{\hbar u}2}+e^{-\frac{\hbar u}2}}{4 u\mathcal{S}(u\hbar)} e^{u(\mathcal{S}(\hbar u z\partial_{z})-1) 
			R(z) } \right)
		\Big|'_{t = z^{-\deg P\deg R}} \Big|''_{u = z^{-\deg R}},
	\end{align}
	where by $|''$ we mean that we only select the terms with $\deg u \geq 0$. The latter expression is manifestly regular at $z\to \infty$. 
	
	Finally, extending our arguments to all variables $z_1,\dots,z_n$, we obtain that $H_{g,n}(z_1,\dots,z_n)$, $2g-2+n>0$, is a rational function that in each of its variables has poles only at the points $p_1,\dots,p_N$ with the odd principal parts with respect to the corresponding deck transformations. This immediately implies that $H_{g,n}\in \Theta_n$. 
\end{proof}

\subsection{Giacchetto--Kramer--Lewa\'nski conjecture and its generalization}\label{S7.3}
Consider the $n$-functions $H_{g,n}$ constructed from Orlov's hypergeometric BKP tau-functions for $\overline \psi=\frac 12 {\mathcal S}(\hbar \partial_y) P(y)$ and $\overline y =y =R(z)$, where $P$ is an arbitrary even polynomial in $y$ and $R$ is an arbitrary odd polynomial in $z$. Recall $X=X(z)=z\exp(-P(R(z)))$. Recall that we defined 
$W_{g,n} = D_1\cdots D_n H_{g,n}$, and we set
\begin{align}\label{eq:rescaling}
	\omega_{g,n}(z_{\llbracket n \rrbracket}) \coloneqq 2^{1-g} W_{g,n}(X_{\llbracket n \rrbracket}) \prod_{i=1}^n \frac {dX_i}{X_i} + \delta_{g,0}\delta_{n,2} \frac{1}{2} B(X_1,X_2) d\log X_1d\log X_2,
\end{align}
With this assignment, it follows from Propositions~\ref{prop:W02} and~\ref{prop:ClosedW01} that 
\begin{align}
	\omega_{0,1}(z) = y\,d\log X \qquad \text{and} \qquad \omega_{0,2}(z_1,z_2) = \mathcal{B}(z_1,z_2).
\end{align}
For all other $\omega_{g,n}$, $g\geq 0$, $n\geq 1$, $2g-2+n>0$, we have the following theorem

\begin{theorem}[Generalized Giacchetto--Kramer--Lewa\'nski conjecture] The symmetric $n$-dif\-fe\-ren\-tials $\omega_{g,n}$ are obtained by the odd topological recursion~\eqref{eq:TopologicalRecursion} for the initial data $X=z\exp(-P(R(z)))$ and $y=R(z)$. 
\end{theorem}

\begin{proof} According to Lemma~\ref{lem:EquiavelentTR} we have to check the blobbed topological recursion and the projection property. The blobbed topological recursion follows from Theorem~\ref{thm:Blobbed}, which is proved in a much more general situation (it is obvious that the analytic assumptions listed in Section~\ref{sec:Assumptions} are satisfied). 
	On the other hand, the linear loop equations and the projection property are equivalent to the statement of Proposition~\ref{prop:QuasiPol}.
\end{proof}

\bibliographystyle{alphaurl}
\bibliography{KPTRref}

\end{document}